\documentclass[11pt,letterpaper,dvipsnames]{article}

\usepackage{verbatim}
\usepackage{url,hyperref}
\usepackage{amsmath,amssymb,amsthm,mathrsfs}
\usepackage{parskip}
\usepackage{esvect}
\usepackage{xcolor}

\newcommand{\ignore}[1]{}

\usepackage[T1]{fontenc}

\usepackage[utf8]{inputenc}

\usepackage{tikz}
\usetikzlibrary{snakes,arrows,shapes}
\usetikzlibrary{shapes.geometric, arrows, arrows.meta, calc, decorations.markings}

\usepackage{fullpage}
\usepackage[margin=1in]{geometry}
\usepackage{bbm}
\usepackage[english]{babel}
\usepackage{yhmath}
\usepackage{stmaryrd}
\usepackage{csquotes}
\usepackage{bm}

\usepackage{hyperref}
\usepackage[vcentermath]{youngtab}
\usepackage{ytableau}\ytableausetup{centertableaux}
\usepackage{amsmath}
\usepackage{amssymb}
\usepackage{amsfonts}
\usepackage{amsthm}
\usepackage{longtable}
\usepackage{tikz}
\usetikzlibrary{calc}
\usetikzlibrary{arrows.meta}

\tikzstyle{vertex}=[circle, draw, inner sep=0pt, minimum width=11pt]
\usepackage{mathtools}
\usepackage{mathrsfs}
\usepackage{physics}
\usepackage{todonotes}

\usepackage{comment}
\usepackage{esvect}

\usepackage[nameinlink]{cleveref}

\theoremstyle{plain}
\newtheorem{theorem}{Theorem}[section]

\theoremstyle{plain}
\newtheorem{proposition}[theorem]{Proposition}
\newtheorem{corollary}[theorem]{Corollary}
\newtheorem{lemma}[theorem]{Lemma}
\newtheorem{fact}[theorem]{Fact}

\newtheorem{question}[theorem]{Question}

\theoremstyle{definition}
\newtheorem{remark}[theorem]{Remark}
\newtheorem{definition}[theorem]{Definition}

\newtheorem{claim}[theorem]{Claim}

\usepackage{xcolor}
\hypersetup{colorlinks,
    linkcolor={blue!80!black},
    citecolor={green!50!black},
    urlcolor={red!40!black}
}

\newcommand{\F}{\mathbb{F}}
\newcommand{\VF}{{\rm VF}}

\newcommand{\K}{\mathbb{K}}
\newcommand{\N}{\mathbb{N}}

\newcommand{\size}{{\rm size}}

\newcommand{\ch}{\mathsf{CH}_{\mathsf{lin}}}

\newcommand{\poly}{{\rm poly}}

\newcommand{\NP}{\mathsf{NP}}
\newcommand{\DP}{\mathsf{P}}

\newcommand{\sharpP}{\mathsf{\#P}}

\newcommand{\VP}{\mathsf{VP}}
\newcommand{\VNP}{\mathsf{VNP}}

\newcommand{\VFPT}{\mathsf{VFPT}}
\newcommand{\VW}{\mathsf{VW}}

\newcommand{\VWnb}{\VW_{\textup{nb}}}

\newcommand{\SUBEXP}{\mathsf{SUBEXP}}
\newcommand{\per}{\operatorname{per}}
\newcommand{\CH}{\mathsf{CH}}

\newcommand{\CHlin}{\mathsf{CH}_\mathrm{lin}}

\newcommand{\DTime}{\mathsf{DTime}}

\newcommand{\pbounded}{p\text{-}\mathsf{bounded}}
\newcommand{\pfamily}{p\text{-}\mathsf{family}}
\newcommand{\pfamilies}{p\text{-}\mathsf{families}}

\newcommand{\ones}[2]{\genfrac{\langle}{\rangle}{0pt}{}{#1}{#2}}

\newcommand{\X}{\mathbf{X}}
\newcommand{\Y}{\mathbf{Y}}
\newcommand{\bZ}{\mathbf{Z}}

\newcommand{\Kk}{\mathbf{K}}

\newcommand{\bit}{\mathsf{bit}}
\newcommand{\Pp}{\mathsf{P}}
\newcommand{\tc}{\mathsf{TC}}
\newcommand{\Zz}{\mathbf{Z}}
\newcommand{\Z}{\mathbb Z}
\newcommand{\cc}{\mathbf{C}}
\newcommand{\sgn}{\mathsf{sgn}}
\newcommand{\J}{\mathbf{J}}
\newcommand{\Nn}{\mathbf{N}}

\newcommand{\CR}{\mathsf{CR}}
\newcommand{\Bit}{\mathsf{Bit}}

\newcommand{\C}{\mathbb C}
\newcommand{\clink}{\mathsf{C\text{-}lin}}
\newcommand{\nch}{\mathsf{CH}}
\newcommand{\nchk}{\mathsf{C}}
\newcommand{\PH}{\mathsf{PH}}

\newcommand{\plsum}{\mathsf{p\text{-}{{\log}\text{-}Expsum}}}

\newcommand{\fpt}{\textup{fpt}}
\newcommand{\fptnb}{\textup{fpt}_{\textup{nb}}}
\newcommand{\form}{\textup{formal-deg}}

\newcommand{\wt}{\textup{wt}}

\newcommand{\VPnb}{\VP_{\textup{nb}}}
\newcommand{\VNPnb}{\VNP_{\textup{nb}}}
\newcommand{\VFPTnb}{\VFPT_{\textup{nb}}}

\newcommand{\arithmetize}{\mathsf{arithmetize}}

\newcommand{\sfF}{\mathsf{F}}

\usepackage{amsthm}

\makeatletter
\def\thm@space@setup{%
  \thm@preskip=8pt plus 2pt minus 2pt
  \thm@postskip=8pt plus 2pt minus 2pt
}
\makeatother

\title{Exponential lower bound via exponential sums\thanks{This is an extended version of the ICALP 2024 paper \cite{DBLP:conf/icalp/BhattacharjeeBD24}.}}
\author{Markus Bl\"aser \thanks{Saarland University, Saarbr\"ucken, Germany; Email:\texttt{mblaeser@cs.uni-saarland.de}} \and Somnath Bhattacharjee \thanks{University of Toronto, Canada; Email: \texttt{somnath.bhattacharjee@mail.utoronto.ca}} \and Pranjal Dutta \thanks{Nanyang Technological University, Singapore; Email: \texttt{duttpranjal@gmail.com}} \and Saswata Mukherjee \thanks{National University of Singapore, Singapore; Email: \texttt{saswatamukherjee607@gmail.com}}}
\date{}
\begin{document}
\maketitle

\begin{abstract}

Valiant's famous $\VP$ vs.\ $\VNP$ conjecture states that the symbolic permanent polynomial does not have polynomial-size algebraic circuits. However, the best upper bound on the size of the circuits computing the permanent is exponential. Informally, $\VNP$ is an exponential sum of $\VP$-circuits. 
In this paper we study whether, in general,
exponential sums (of algebraic circuits) {\em require} exponential-size algebraic circuits. We show that the famous Shub-Smale $\tau$-conjecture indeed implies such an exponential lower bound for an exponential sum.
Our main tools come from parameterized complexity. 
%Of course,
%all our lower bounds will only be conditional.
Along the way, we also prove an exponential fpt (fixed-parameter tractable) lower bound for the parameterized algebraic complexity class $\VWnb^0[\Pp]$, assuming the same conjecture. $\VWnb^0[\Pp]$ can be thought of as the weighted sums of (unbounded-degree) circuits, where only $\pm 1$ constants are {\em cost-free}. To the best of our knowledge, this is the {\em first} time the Shub-Smale $\tau$-conjecture has been applied to prove explicit exponential lower bounds.

Furthermore, we prove that when this class is fpt,
then a variant of the counting hierarchy, namely the {\em linear counting hierarchy} collapses.
Moreover, if a certain type of parameterized exponential sums is fpt, then integers, as well as polynomials with coefficients being {\em definable} in the linear
counting hierarchy have subpolynomial $\tau$-complexity.

Finally, we characterize a related class $\VW[\sfF]$, in terms
of permanents, where we consider an exponential sum of algebraic formulas instead of circuits. We show that when we sum over cycle covers that have one long cycle and all other cycles have constant length,  then the resulting family of polynomials is {\em complete} for $\VW[\sfF]$ on certain types of graphs.
\end{abstract}

\tableofcontents
\section{Introduction}
\label{intro}

Valiant \cite{DBLP:conf/stoc/Valiant79a} proposed an algebraic version of the $\DP$ versus $\NP$ question and defined the class $\VP$, the algebraic analogue of $\Pp$, which contains polynomial families computable by polynomial
sized algebraic circuits. An \textit{algebraic circuit} (or, arithmetic circuit) $C$ is a directed acyclic graph such that 
\begin{enumerate}
    \item  every node has either in-degree ({\em fan-in}) $0$ (the {\em input gates}) or $2$ (the {\em computational gates}), 
    \item every input gate is labeled by elements from a field $\K$ or variables from $\X=\{X_1, \cdots, X_n\}$, 
    \item every computational gate is labeled by either $+$ (addition gate) or $\times$ (multiplication gate), with the obvious syntactic meaning, and
    \item  there is a unique gate of
out-degree $0$, the {\em output gate}.
\end{enumerate}
Clearly, every gate in a circuit computes a polynomial in $\K[\X]$. We say that the circuit $C$ computes $P(\X)\in \K[\X]$ if the output gate of $C$ computes $P(\X)$. The {\em size} of $C$, denoted by $\size(C)$, is the number of nodes in the circuit. An algebraic circuit is an \textit{algebraic formula} if every gate in the circuit has out-degree $1$ except for the output gate. The class $\VNP$, the algebraic analogue of $\NP$, is definable 
by taking {\em exponential sums} of the form
\begin{equation}\label{eq:expsum}
  f(\X) = \sum_{e \in \{0,1\}^\ell} g(\X,e)\;,
\end{equation}
where $g$ is computable by a polynomial-size circuit and $\ell$ is polynomial in the number of variables. It is known that one can also replace algebraic circuits by algebraic formulas, and still get the same class $\VNP$ \cite{DBLP:conf/stoc/Valiant79a,DBLP:journals/jc/MalodP08}. Valiant further proved that the permanent family is complete for $\VNP$ (over fields of characteristic
not two).
Recall that the permanent of a matrix $(X_{i,j})$ is defined as
\begin{equation}\label{eq:per}
  \per \X = \sum_{\pi \in S_n} X_{1,\pi(1)} \cdots X_{n,\pi(n)}.
\end{equation}
The famous Valiant's conjecture $\VP \not= \VNP$ is equivalent to the
fact that the permanent does not have polynomial-size circuits. 
The representation of the permanent in (\ref{eq:per}), although it looks very natural,
is not {\em optimal}. Ryser's formula \cite{ryser} yields an algebraic formula of
size $O(2^n n^2)$. A formula of similar size was later found by Glynn \cite{GLYNN20101887}.
Ryser's formula is now over sixty years old and has not been improved since.
This gives rise to the interesting question whether there is a formula or circuit
of subexponential-size (in $n$) for the permanent? More generally, we can now ask the following question.

\medskip
\begin{question}\label{qn-1}
Is an exponential sum $f$ (as in~Eq.~(\ref{eq:expsum}))
always computable by an algebraic circuit or formulas
of size subexponential in $\ell$, that is, size $2^{o(\ell)}$? 
Or are there instances for which exponential-size is necessary?
\end{question}
Note that exponential-size being necessary is a much 
{\em stronger} claim than $\VP \not= \VNP$.
It could well be that $\VP \not= \VNP$ but still exponential sums
like in (\ref{eq:expsum}) have subexponential size circuits!
In this paper, we shed some light on the question
what happens if exponential sums would always have
{\em subexponential} size circuits.

Question~\ref{qn-1} serves as a driving force between the famous Shub-Smale $\tau$-conjecture \cite{shub1995intractability} and {\em exponential} lower bounds on exponential sums. The $\tau$-complexity $\tau(f)$ of an integer polynomial is the size of a smallest division-free circuit that computes $f$ starting from the constants $\pm 1$.
The \emph{$\tau$-conjecture} states that the number of integer zeroes of $f$
is polynomially bounded in $\tau(f)$, see \cite{shub1995intractability}. \cite{shub1995intractability}
shows that the $\tau$-conjecture implies $\Pp_{\C} \ne \NP_{\C}$, in the Blum–Shub–Smale (BSS) model
of computation over the complex numbers \cite{blum1989theory,blum2000algebraic}. 
%A BSS machine is a Random Access Machine (RAM) with registers that can store arbitrary real numbers and
%compute rational functions over reals in a single time step. Thus, the BSS machines are more powerful than Turing
%machines.
\smallskip

{\bf \noindent Super-polynomial lower bounds assuming the~$\tau$-conjecture.}~B\"urgisser \cite{article}
connected the $\tau$-complexity of the permanent to various other conjectures. He showed that the $\tau$-conjecture implies a {\em superpolynomial} lower bound on $\tau(\per_n)$, implying the constant-free version of $\VP \ne \VNP$, namely~$\VP^0 \ne \VNP^0$; for definitions, see~\Cref{const-free}. The proof strategy of~\cite{article} is as follows: assume $\tau(\per_n) = \poly(n)$, and conclude a complexity-theoretic `collapse' that the counting hierarchy $\CH$ (for a definition, see~\Cref{sec:linear-ch}) is in $\Pp/\poly$. Consider the Pochhammer–Wilkinson polynomial $f_n(x) := \prod_{i=1}^n (x-i)$, and construct a unique $O(\log n)$-variate multilinear polynomial $B_n$ such that under a `suitable' substitution, one gets back $f_n$. The coefficients of $f_n$ as well as $B_n$, are efficiently computable (since $\CH \subseteq \Pp/\poly$), implying $B_n \in \VNP^0$. An inspection of Valiant’s completeness result reveals that if $B_n \in \VNP^0$, then there is a polynomially bounded sequence $p(n)$ such that $\tau(2^{p(n)}B_n)=\poly(\log n)$, which implies $\tau(2^{p(n)}f_n)=\poly(\log n)$, contradicting the $\tau$-conjecture. 

In~\cite{article}, the superpolynomial lower bound on $\tau(\per_n)$ was also implied by any of the quantities 
$\tau(n!)$, $\tau(\sum_{k = 0}^n \frac{1}{k!}T^k)$, or $\tau(\sum_{k = 0}^n k^rT^k)$
(for any fixed negative integer $r$) not being poly-logarithmically bounded as a function of $n$. Here, we remark that the separation proof of $\VP^0$ and $\VNP^0$, even assuming {\em strong} bounds on the $\tau$-conjecture, is merely {\em superpolynomial}: we {\em do not} get the (possibly) desirable exponential separation between $\VP^0$ and $\VNP^0$. This leads to the following question. 
%Here, we stress on the fact that a {\em weaker assumption} of the permanent polynomial having a subexponential-size constant-free circuit {\em does not} give  any explicit lower bound on $\tau(\per_n)$. The main reason being: collapse of counting hierarchy to subexponential regime ($\CH \subset \SUBEXP/\poly$) {\em doesnot} show that $B_n \in \VNP^0$. Similarly, just This leads to the following question. \pdc{Needs a different flow from subexp assumption to exp lb.}
\begin{question}\label{qn-2}
    Does the~$\tau$-conjecture imply exponential algebraic lower bounds?
\end{question}
Here, we mention that there are variants of the $\tau$-conjecture, e.g.,~the {\em real $\tau$-conjecture}~\cite{DBLP:conf/innovations/Koiran11,tavenas2014bornes} or {\em SOS-$\tau$-conjecture}~\cite{dutta2021real}, which also give strong algebraic lower bounds. There is also super polynomial lower bound known from a proof complexity theoretic view due to \cite{DBLP:journals/corr/abs-1911-06738} from the original Shub-Smale $\tau$-conjecture. However, the Shub-Smale $\tau$-conjecture is {\em not known} to give an exponential lower bound for the permanent. 

%\pdc{Edited till this..}

\subsection{Our results}

The results of our paper revolve around answering both \Cref{qn-1}-\ref{qn-2} positively. The main result is the following. 

\medskip
\begin{theorem}[Informal version of~\Cref{thm:formal-exp-lb}]\label{thm:informal-1}
The $\tau$-conjecture implies an exponential lower bound for 
some explicit exponential sum.
\end{theorem}

{\noindent \bf Remarks.}~(1) Although the existence of {\em some} polynomial requiring exponential circuits is clear from dimension/counting, the existence of an (even non-explicit) exponential sum polynomial requiring exponential-size circuits is {\em unclear}. Explicit here means that the family is in $\VNP$. 

(2) One can also think of an exponential sum $f$ in~\Cref{eq:expsum}, as~$f= \sum_{e \in \{0,1\}^{\ell(n)}} U(\X,y,e)$, where $U(\X,\Y,\mathbf{Z})$ is a {\em universal circuit} of size $\size(U) = \poly(\size(g))$ with $\Y= (Y_1, \hdots, Y_r)$ and $\bZ = (Z_1, \hdots, Z_{\ell(n)})$ and $y \in \F^{r}$ is chosen such that $U(\X,y,e)=g(\X,e)$; and the number of variables $\ell(n)$ is linear in $n$. %We also point out that the above result {\em does not} give any strong exponential lower bound for the permanent, mainly because the universal circuit will have to simulate a polynomial (but fixed) number of steps.

(3) Since there is a polynomial (non-linear)  blowup in the reduction of the exponential sum on the universal circuit from the permanent, we will only get a subexponential lower bound on the permanent polynomial assuming the $\tau$-conjecture. We leave it as an open question to achieve an exponential lower bound on the permanent assuming the $\tau$-conjecture.

%\pdc{Probably should say that we get something for this? But this doesn't give permanent lb?}

%\pdc{We can also remove this. I have said something very informally in Section 6. Maybe that suffices. But we should still say that this does not give exp lb for permanents.}
%\end{remark*} %In particular,  it suffices to argue exponential lower bounds (if they exist) against the exponential sum of universal circuits. Note that the number of variables $\ell(n)$ that we sum over
%is linear in the number $n$ of input variables.

\medskip
The proof of \Cref{thm:informal-1} is rather indirect, and goes via {\em exponential sums}, which is our main object of study (and bridge between many results and classes). 

\medskip
{\bf \noindent $\log$-variate exponential sum polynomial.}~ Let $g(\X,\Y)$ be some
polynomial in $n$-many $\X$-variables and $\ell(n)$-many $\Y$-variables, where
$\ell(n) = O(n)$. Assume that $g$ is computed by a circuit of size $m$. Then
we define
\[
   \plsum_{m,k}(g) \;:=\;\sum_{y\in\{0,1\}^{\ell(n)}}\,g(\X,y)\;,
\]
where $k = n/\log m$. The size of the exponential is measured in the number $\ell(n)$ of $\Y$-variables.
In the end, we want to measure in the input size, the number $n$ of $\X$-variables.
To talk about subexponential complexity, $\ell(n)$ should be linearly bounded.
$g$ will be typically computed by a circuit (of unbounded degree).
We want to view $\plsum_{m,k}$ as a parameterized problem, the parameter will be $k$. Our definition of $\plsum$, as a polynomial-sum, is motivated by the {\em log-parameterizations} which are used in the definition of the so-called $M$-hierarchy in the Boolean setting, see \cite{DBLP:journals/eatcs/FlumG04, DBLP:series/txtcs/FlumG06}.

We show that $\plsum$ is most likely {\em not} fixed-parameter tractable (fpt). A polynomial family~$p_{n,k}$ is fpt if both its size and degree are fpt bounded, i.e.,~of the form $f(k)q(n)$, for $q \le \poly(n)$, and $f : \N \to \N$ being {\em any} computable function.  We connect $\plsum$ with
\begin{enumerate}
    \item a linear variant of the counting hierarchy (we denote it by $\CHlin$), where the size of the oracle calls are bounded linearly in the size of the input; for definition see Section~\ref{sec:linear-ch}; and
    \item integers definable in $\CHlin$, similar to B\"urgisser \cite{article}. Informally, an integer is definable in $\CHlin$, if its sign and bits are computable in the same class.
\end{enumerate}

%Under these premises, we prove the following in the parameterized setup; for a formal statement, see~Theorem~\ref{thm:tau}. Subsequently, we answer \Cref{qn-1} in~\Cref{remark:exp-sum-exp-size}.

%In Section~\ref{sec:linear-ch}, we define 
%This turns out to be important when dealing with subexponential complexity. 

%here are tight connections known  already between
%parameterized and subexponential complexity in the Boolean setting, 
%see e.g.~\cite{DBLP:journals/eatcs/FlumG04, DBLP:series/txtcs/FlumG06} for an overview. 
%In particular, Here we remark that throughout the paper, we will assume constant-free different restrictions on $g$ (e.g.,~constant-free, bounded degree), which would be clear from the context.

%Similar to B\"urgisser \cite{article}, we define the concept of integers definable in the 
%\emph{linear} counting hierarchy. An integer sequence $a(n,k)$ 
%is definable in the linear counting hierarchy
%if the languages $\sgn(a)$ and $\Bit(a)$ are both in the linear counting hierarchy.
%It turns out the integer sequences definable in the linear counting hierarchy
%share similar closure properties. This is due to the fact that all the closure properties
%proved by B\"urgisser stem problem dlogtime-uniform $\TC^0$-circuits. This will ensure
%that all resulting oracle queries are linearly bounded.

%Under these premises, we prove the following in the parameterized setup; for a formal statement, see~Theorem~\ref{thm:tau}. Subsequently, we answer \Cref{qn-1} in~\Cref{remark:exp-sum-exp-size}.

\begin{theorem}[Informal version of \Cref{thm:collpase-chsubexp} and \Cref{thm:tau}] \label{thenewtheoremtwo}
    If $\plsum$ is fixed-parameter tractable, then the following results hold.
    \begin{enumerate}
        \item The linear counting hierarchy ($\CHlin$) collapses.
        \item Any sequence $a(n)$ definable in the linear counting hierarchy, as well as univariate polynomials with coefficients being definable in the linear counting hierarchy, have subpolynomial $\tau$-complexity.
    \end{enumerate}
\end{theorem}
Finally, many algebraic complexity classes can be defined in terms of permanents.
Most prominently, the ``regular'' permanent family $(\per_n)$ is complete for 
$\VNP$. The class $\VW[1]$ is an important class in parameterized complexity.
It is defined as a bounded sum over constant depth weft-1 circuits.
Bounded sum means that we sum over $\{0,1\}$-vectors with $k$ ones and $k$ is the parameter.
Bl\"aser and Engels \cite{blser_et_al:LIPIcs:2019:11464} prove that $\VW[1]$ is described by so-called $k$-permanents with $k$ being the parameter. 
In a $k$-permanent, we only sum over permutations with $n-k$ self-loops. 

The crucial parameterized class of this work is $\VW[\Pp]$:
it is defined as a bounded exponential sum over polynomially-sized arithmetic circuits computing a polynomial of degree that is polynomially bounded.
While we do not characterise $\VW[\Pp]$ in terms of permanents, we
%We do not know whether we can characterize the class $\VWnb[\Pp]$, which
%is the most relevant for this work, in terms of permanents.
characterize the related class $\VW[\sfF]$: Here instead of summing over circuits, we sum over {\em formulas}.\footnote{Maybe an explanation of the naming
convention is helpful: In $\VW[\Pp]$, we sum of polynomial-size circuits,
which describe the class $\VP$. In $\VW[\sfF]$, we sum over polynomial
size formulas, which define the class $\VF$, the modern name for $\VP_e$.} 
The permutations that we sum over 
for defining our permanent family
will have one cycle of length $k$ and all other cycles bounded by $4$.
Again, $k$ is the parameter.
We call the corresponding polynomials $(k,4)$-restricted permanents. 

It turns out that we also need to restrict the graph classes.
We call a graph $G = (V,E)$ $(4,b)$-nice if we can partition the
set $V = V_1 \cup V_2$ disjointly, such that in the induced graph $G[V_1]$,
every cycle is either a self-loop or has length $> 4$ and
in the induced graph $G[V_2]$ has tree-width bounded by $b$.
While this looks artificial at a first glance, it turns out
that there is a constant $b$ such that $(k,4)$-restricted permanent on $(4,b)$-nice graphs describes the natural class $\VW[\sfF]$.
There is a family of $(4,b)$-nice graphs such that the corresponding
family of $(k,4)$-restricted permanents is $\VW[\sfF]$-hard. On the other hand, the $(k,4)$-restricted permanent family  is in $\VW[\sfF]$
for every family of $(4,b)$-nice graphs. Together, this implies: 
\begin{theorem}[{$\VW[\sfF]$}-Completeness]\label{thm:informal-3}
$(k,4)$-restricted permanent family on $(4,b)$-nice graphs are $\VW[\sfF]$-complete.
\end{theorem}
We also prove strong separations of algebraic complexity classes and parameterized algebraic complexity classes (\Cref{vw-collapse}), and exponential lower bound in the parameterized setting (\Cref{thm:converse}).

\noindent For $\VNP$ it is known that it does not matter whether we sum over 
formulas or circuits, that is, $\VNP = \VNP_e$. Whether 
$\VW[\Pp] = \VW[\sfF]$ remains an open questions for future research.

\subsection{Proof ideas}

%\pdc{I have tried removing VWnb etc; also rearranged things a bit; please read the proof-idea section once again}
In this section, we briefly sketch the proof ideas. The omitted proofs of the paper can be found in the longer arxiv version of the paper. We first present the proofs of \Cref{thenewtheoremtwo}, because the techniques and lemmas involved in proving them are the backbone of \Cref{thm:informal-1}.

\smallskip
{\bf \noindent Proof idea of \Cref{thenewtheoremtwo}.} We prove them in two parts.

{\noindent \em Proof of Part (1):} We prove even a stronger statement for the subexponential version of the linear counting hierarchy. The proof goes via induction on the level of the counting hierarchy. The criteria for some language $B$ being in the $(k+1)$-th level are that there should be some language $A$ in the $k$-th level such that $|\{y\in\{0,1\}^n:\langle x,y\rangle\in A\}|>2^{n-1}$.  Essentially, for a language $A$ in the $k$-th level, we express  $|\{y \in \{0,1\}^n : \langle x,y \rangle \in A\}| > 2^{n-1}$ as an exponential sum over an algebraic circuit $\chi_A(x,y)$, which captures the characteristic function of $A$. Furthermore, one can show that $\plsum$ is fpt (in an unbounded constant-free setting) iff $\sum_{y} g(\X,y)$ has $2^{o(n)} \poly(m)$ size circuits, where $g$ has a circuit of size $m$; see~\Cref{thm:subexp:1} and~\ref{thm:converse:alt}. Putting these together, one gets that the exponential sum has a subexponential-size constant-free circuit. Lastly, we want to get the information about the highest bit of the sum (which is equivalent to looking at it $\bmod~2^n$), which can be efficiently {\em arithmetized}. In every step there is polynomial blowup in the size, and hence the size remains subexponential, yielding the desired result. For details, see \Cref{thm:collpase-chsubexp}.
%~\Cref{sec:appendix-sec4}. 

{\em \noindent Proof of Part (2):} This proof is an adaption of~\cite{article,koiran2005valiant} in our context. Take a sequence $(a_n)_n \in \ch\Pp$. We define a multilinear polynomial $A(\Y)$ such that the coefficient of $\Y^{\mathbf{j}}$ is the $j$-th bit of $a(n)$, where $\mathbf{j}$ is the binary representation of $j$. Furthermore, checking $a(n,j)=b$ can be done by a subexponential circuit $C(\Nn,\J)$, where $\Nn$ and $\J$ have $\log n$ and $\bit(n)$-many variables capturing $n$ and $j$ respectively. Moreover, one can define $F(\Nn,\Y,\J) = C(\Nn,\J) \cdot \prod_{i} (J_i Y_i + 1 - J_i)$ and show that $A$ can be expressed as an exponential sum over $F(j,\Nn, \Y)$! This is clearly a $\plsum$ instance, which finally yields that the $\tau$-complexity of $a(n)$ is subpolynomial. A similar proof strategy also holds for the polynomials with coefficients being definable in $\ch\Pp$. For details, see~\Cref{sec:tau}.

%\medskip
%\bf \noindent Proof idea of \Cref{thm:informal-2}.} The key result is Theorem~\ref{plsumthm}, which is essentially a completeness result
%for $\plsum_{m,k}$: Every family in $\VW[\Pp]$ can be written as an fpt-substitution of a bounded sum summing over a polynomial-size circuit. Let $f$ be a $\plsum_{m,k}$ instance, i.e.,~$f=\sum_{y} g(\X,y)$, where $g(\X,\Y)$ has a constant-free circuit of size $m$. The main idea here is to partition the variables $\Y$ into $k=n/\log m$ sets $E_1, \cdots, E_k$, and transform $g$ into $\tilde{g}$ such that the final summation is over $k$-weight integers. To do so, for each $S\subseteq E_i$ we take a {\em new variable} $Z_i^S$. We do this for all $i$. Define $\overline{Z_i} :=\{Z_i^S \,:\, S \subseteq E_i\}$ and $\bZ=\bigcup_i\overline{Z_i}$. We call an assignment to $\Z$ a  {\em good assignment} if exactly one variable in each $\overline{Z_i}$ gets the value $1$. There is a {\em one-to-one} correspondence between
%    $\{0,1\}$ assignments to the $\Y$ variables and {\em good assignments} to the $\bZ$ variables: think of this as substituting $\varphi: Y_i \mapsto \prod_{S\subseteq E_i,\ Y_j\not\in S}(1-Z_i^S)$. This substitution gives $\tilde{g}$, computed by a small-size circuit, and more importantly the correspondence helps us to write $f$ as sum over $k$-weight assignments.

{\bf \noindent Proof idea of \Cref{thm:informal-1}}~Take the Pochhammer polynomial $p_n(X) = \prod_{i = 1}^{n} (X+i)$. The coefficient of $X^{n-k}$ in $p_n$ will be $\sigma_k(1,\dots,n)$, where $\sigma_k(z_1,\dots,z_n)$ is the $k$-th elementary symmetric polynomial in variables $z_1,\dots,z_n$. One can show that $\ch\Pp$ is closed under polynomially-many additions and multiplications~(\Cref{aaa}). Therefore, $(\sigma_k(1,\dots,n))_{n\in\N,k\leq n}$ is definable in the linear counting hierarchy (see~\Cref{ch}). And by \Cref{thm:tau}, $(p_n)_{n\in\N}$ has $n^{o(1)}$-sized constant-free circuits if $\plsum$ is fixed-parameter tractable. But $p_n$ has $n$ distinct integer roots. Assuming the $\tau$-conjecture, $\plsum$ is not {\em fpt}. 
On the other hand, one can show that when exponential sums over circuits of size $m$ have circuits
have size $2^{o(n)} \poly(m)$, then the $\plsum$ is fpt, by Theorem~\ref{thm:converse:alt}; in other words, $\plsum$ is not fpt implies an exponential lower bound on an exponential sum. This finishes the proof.  %\pdc{Our theorem is different. So need to finish the argument.}

\medskip
{\bf\noindent Proof idea of \Cref{thm:informal-3}.}~The hardness proof is
{\em gadget} based (\Cref{cor:kcper:hard}). The details are however quite complicated since we have to cleverly keep track of the cycle lengths. For the upper bound,
we work along a tree decomposition. While it is known that the permanent can be computed in fpt time on graphs of bounded
treewidth, we cannot simply adapt these algorithms, since we
have to produce a formula. This can be achieved using a {\em balanced tree decomposition}.

%\pdc{Maybe somewhere we should say that most of the proofs go through a constant-free unbounded setting?}

\subsection{Previous results}

To prove (conditional) exponential lower bounds, the standard assumptions
that $\Pp \not= \NP$ or $\VP \not= \VNP$ are not enough. It is
consistent with our current knowledge that for instance $\DP \not= \NP$,
but $\NP$-hard problems can have subexponential time algorithms. 
What we need is a complexity assumption
stating that certain problems can only be solved in exponential time.
This is the exponential time hypothesis (ETH) in the Boolean setting. 
Dell et al.~\cite{DBLP:journals/talg/DellHMTW14} studied the exponential time complexity of the permanent,
they prove that when there is an algorithm for computing the permanent in time
$2^{o(n)}$, then this violates the counting version of the exponential time hypothesis
\#ETH. \#ETH states that there is a constant $c$ such that no deterministic algorithm 
can count the number of satisfying assignments of a formula in $3$-CNF in time
$2^{cn}$. For connections between
parameterized and subexponential complexity in the Boolean setting, we refer to~\cite{DBLP:journals/eatcs/FlumG04, DBLP:series/txtcs/FlumG06}.

Bl\"aser and Engels \cite{blser_et_al:LIPIcs:2019:11464} transfer the important definitions and results from parameterized complexity in the Boolean world to define a theory of parameterized algebraic complexity classes. In particular, they define the $\VW$-hierarchy and prove that the
clique polynomial and the $k$-permanent are $\VW[1]$-complete (under so-called
fpt-substitutions). They also claim the hardness of the restricted
permanent for the class $\VW[t]$ for every constant $t$ and sketch a proof. 
Note that $\VW[\sfF]$ contains each $\VW[t]$. So we strengthen the hardness proof
in \cite{blser_et_al:LIPIcs:2019:11464} and complement it with an upper bound.

The main tool used by B\"urgisser~\cite{article} to prove the results above is the counting hierarchy.
The polynomial counting hierarchy was introduced by Wagner \cite{DBLP:journals/acta/Wagner86} to classify the complexity of Boolean counting problems. 
The fact that small circuits for the
permanent collapses the counting hierarchy is used by B\"urgisser to prove the
results mentioned above. 

Finally, there have been quite a few works~\cite{article,koiran2005valiant,koiran2011interpolation,DBLP:conf/innovations/Koiran11}, where we have conditional separations on the constant-free version of $\VP$ and $\VNP$, namely~$\VP^0$ and $\VNP^0$, or their variants, depending on the strength of the conjecture. But this is the first time that we are separating algebraic classes and proving exponential lower bounds, assuming the $\tau$-conjecture.

\subsection{Structure of the paper}

% From Section 3 on so far
%In Section~\ref{vw-hard} we prove some easy conditional\ collapse results
%of the $\VW$-hierarchy in various circuit models. 
%In Section~\ref{sec:val-CH}, we connect Valiant's model to the counting
%hierarchy. We introduce exponential sums and investigate its relation to
%the parameterized classes. Here, the main result is that the fixed-parameter tractability
%of exponential sums collapses the counting hierarchy.
%Section~\ref{sec5} introduces integers in the linear counting hierarchy.
%The proofs are quite similar to \cite{article}, however we need to
%pay special attention to the fact the witness size is linear.
%In Section~\ref{sec:tau}, we make the connection to the $\tau$-conjecture.
%Finally, in Section~\ref{sec:complete}, we prove the completeness
%of the restricted permanent. Due to space limitations,
%many proofs had to be omitted. They can be found in the appendix. \pdc{Check once about the structure}

In \Cref{sec:prelim-1}, we defined the basics of constant-free Valiant's model and the unbounded and parameterized setting. In~\Cref{sec:linear-ch}, we introduce the linear counting hierarchy ($\CHlin$) and its basic properties. \Cref{sec:val-CH} connects Valiant's model to the counting
hierarchy. Here, we formally introduce exponential sums and investigate their relation to
the parameterized classes. The main result is that the fixed-parameter tractability
of exponential sums collapses the counting hierarchy. The proofs are quite similar to \cite{article}, however, we need to
pay special attention to the fact the witness size is linear. \Cref{sec5} introduces the definability (computability) of integers in the linear counting hierarchy, and some closure properties of the same. \Cref{sec:tau} proves the exponential lower bound on exponential sum assuming $\tau$-conjecture. \Cref{sec:VW} introduces the parameterized $\VW$-classes and its basic properties. In Section~\ref{vw-hard} we prove some easy conditional\ collapse results
of the $\VW$-hierarchy in various circuit models.  %\pdc{Wrote the new structure. Please go through and edit.}
%\pdc{Keep the comments; and don't need to necessarily change whatever being suggested. Could discuss later and decide.}
%\mdc{Modified it a bit. Looks good now.}
\section{Preliminaries I}
\label{sec:prelim-1}

\subsection{Constant-free and unbounded models} \label{const-free}

\textbf{Constant-free Valiant's classes:} \label{constant-free-valiant}
We will say that an algebraic circuit is \emph{constant-free}, if no field elements other than $\{-1,0,1\}$ are used for labeling in the circuit. Clearly, constant-free circuits can {\em only} compute polynomials in $\Z[\X]$. For $f(X)\in\Z[\X]$, $\tau(f)$ is the size of a minimum size constant-free circuit that computes $f$, while $L(f)$ denotes the minimum size circuit that computes $f$. It is noteworthy to observe that, {\em unlike} Valiant's classical models, computing integers in the constant-free model can be costly; e.g.,~$\tau(2^{2^n}X^n)=\Omega(n)$, while $L(2^{2^n}X^n) =\Theta(\log n)$.  On the other hand, for any $f\in\Z[\X]$, $L(f)\leq\tau(f)$. 

Before defining the constant-free Valiant classes, we formalize the notion of {\em formal degree} of a node, denoted $\form(\cdot)$. It is defined recursively as follows: (1) the formal degree of an input gate is $1$, (2) if $u = v+w$, then $\form(u) = \max (\form(v), \form(w))$, and (3) if $u = v \times w$, then $\form(u) = \form(v) + \form(w)$. The
formal degree of a circuit is defined as the formal degree of its output node.

The class {\em constant-free Valiant's $\Pp$}, denoted by~$\VP^0$, contains all $\pfamilies$ $(f)$ in $\Z[\X]$, such that $\form(f)$ and $\tau(f)$ are both $\pbounded$. Analogously, $\VNP^0$ contains all $\pfamilies$ $(f_n)$, such that there exists a $\pbounded$ function $q(n)$ and $(g_n) \in \VP^0$, where~\[
f_n(\X)\;=\;\sum_{\overline{y}\in\{0,1\}^{q(n)}}g_n(\X,y_1,\dots,y_{q(n)})\;.\]

It is not clear whether showing $\VP^0 \ne \VNP^0$ implies $\VP \ne \VNP$, it is {\em not even clear} whether $\VP^0\neq\VNP^0$ $\implies$ $\tau(\per_n) = n^{\omega(1)}$. The {\em subtlety} here is that
in the algebraic completeness proof for the permanent, {\em divisions by two} occur! However, a partial implication is known due to~\cite[Theorem~2.10]{article}: Showing $\tau(2^{p(n)}f_n)=n^{\omega(1)}$, for some $f_n\in\VNP^0$ and all $\pbounded p(n)$ would imply that $\tau(\per_n)=n^{\omega(1)}$.

\medskip
\noindent \textbf{Arithmetization} is a well-known technique in complexity theory.
To arithmetize a Boolean circuit $C$ computing a Boolean function $\varphi$, 
we use the arithmetization technique wherein we map $\varphi(x_1, \hdots, x_n)$ to a polynomial $p(x_1, \hdots, x_n)$ such that for any assignment of Boolean values $v_i \in \{0, 1\}$ to the $x_i$, $\varphi(v_1, \hdots, v_n) = p(v_1, \hdots, v_n)$ holds. 

We define the arithmetization map $\Gamma$ for variables $x_i$, and clauses $c_1, \hdots, c_m$, as follows:
\begin{enumerate}
    \item $x_i \mapsto x_i$,
    \item $\neg x_i \mapsto 1-x_i$,
    \item $c_1 \lor \cdots \lor c_m \mapsto 1 - \prod_{i \in [m]}(1-\Gamma(c_i))$,
    \item $c_1 \land \cdots \land c_m \mapsto \prod_{i \in [m]} \Gamma(c_i)$.
\end{enumerate}
This map allows us to transform $C$ into an arithmetic circuit for $p$.
For a Boolean circuit $C$, we denote the arithmetized circuit by $\arithmetize(C)$. Here, we remark that the degree of $\arithmetize(C)$ can become {\em exponentially} large; this is because there is no known depth-reduction for Boolean circuits, and hence the degree may double at each step, owing to an exponential blowup in the degree. 

\medskip
\noindent \textbf{Valiant's classes in the unbounded setting:}
It is well-known that an algebraic circuit of size $s$, can compute polynomials of degree $\exp(s)$; e.g.,~$f(x)=x^{2^{s}}$, and $L(f) = O(s)$. This brings us to the next definition, the class $\VPnb$, originally defined in~\cite{malod2007complexity}. A sequence of polynomials $(f) = (f_n)_n \in \VPnb$, if the number of variables in $f_n$ and $L(f_n)$ are both $\pbounded$ (the degree {\em may be} exponentially large). The subscript ``nb'' signifies the {\em ``not bounded''} phenomenon on the degree of the polynomial, in contrast to the original class $\VP$. Similarly, a sequence of polynomials $(f) = (f_n)_n \in \VNPnb$, if there exists a $\pbounded$ function $q(n)$ and $g_n(\X,Y_1,\dots,
Y_{q(n)})\in \VPnb$ where~\[
f_n(\X)\;=\;\displaystyle\sum_{\overline{y}\in\{0,1\}^{q(n)}}\,g_n(\X,y_1,\dots,y_{q(n)})\;.\]
One can analogously define $\VPnb^0$ and $\VNPnb^0$, in the constant-free setting. It is obvious that $\VPnb = \VNPnb$ implies $\VP = \VNP$, but the converse is {\em unclear}. However, \cite{malod2007complexity} showed that over a ring of positive characteristic, the converse holds, i.e., $\VP = \VNP$ implies $\VPnb = \VNPnb$! On the other hand, \cite{koiran2011interpolation} showed that $\VP^0=\VNP^0$ implies that $\VPnb^0=\VNPnb^0$, and the converse is unclear because it seems difficult to rule out the possibility that some polynomial family in $\VNP^0$ does not lie in $\VP^0$, but still in $\VP$ (i.e., computable by
polynomial-size algebraic circuits using {\em exponentially large-bit} integers).

\subsection{Parameterized Valiant's classes}
\label{para-val}

Parameterized Valiant's classes were introduced in \cite{blser_et_al:LIPIcs:2019:11464}. 
We will briefly review the definitions and results there and extend them to the constant-free
and unbounded setting.
We first start with the fixed-parameter tractable classes.
The $W$-hierarchies will be introduced later since we only need them in the 
second part of this work.

Our families of polynomials will now have two indices. They will be of the form $(p_{n,k})$. Here, $n$ is the index of the family and $k$ is the parameter. We will say a polynomial family $(p_{n,k})$ is a \textit{parameterized $\pfamily$} if the number of variables is $\pbounded$ in $n$ and the degree is $\pbounded$ in $n,k$. If there is no bound on the degree, we say it is {\em parameterized} family. %\pdc{This is the right definition, right?}

The most natural parameterization is by the degree: Let $(p_n)$ be any $\pfamily$ then we get a parameterized family $(p_{n,k})$ by setting $p_{n,k} :=$ the homogeneous part of degree $k$ of $p_n$. %Since $deg(p_n)$ is polynomially bounded, $p_{n,k}$ is zero when $k$ is large enough. (This will usually be the case for any parameterization.) 
%We will see some of parameterizations in this paper. 
For more details, we will refer the reader to \cite{blser_et_al:LIPIcs:2019:11464}. 
%\pdc{This part is confusing when the degree is unbounded.}

We now define fixed-parameter variants of Valiant’s classes with the constant-free version. \begin{definition}[Algebraic FPT classes]
    \begin{enumerate}
        \item A parameterized $\pfamily$ $(p_{n,k})$ is in $\VFPT$ iff $L(p_{n,k})$ is upper bounded by $f(k)q(n)$ for some $\pbounded$ function $q$ and some function $f:\N\to\N$ (such bound will be called an fpt bound). If one removes the requirement of $\pfamily$ on $p_{n,k}$, and imposes only that the number of variables is $\pbounded$, one gets the class $\VFPTnb$.
        \item A parameterized $\pfamily$ $p_{n,k}$ is in $\VFPT^0$ iff $\tau(p_{n,k})$ is upper bounded by $f(k) q(n)$ for some constant $\pbounded$ function $q$ and some function $f:\N\to\N$. Similarly, one gets $\VFPTnb^0$, if one removes the requirement of $\pfamily$, and imposes only that the number of variables is $\pbounded$.
    \end{enumerate}
\end{definition}
We remark that in the above, $f$ need not be computable as Valiant's model is non-uniform.
\begin{definition}[Fpt-projection]
    A parameterized family $f = (f_{n,k})$ is an fpt-projection of another parameterized family $g = (g_{n,k})$ if there are functions $r, s, t: \N \to \N$ such that $r$ is $\pbounded$, $s,t$ are computable and $f_{n,k}$ is a projection of $g_{r(n)s(k),k'}$ for some $k' \leq t(k)$,\footnote{$k'$ might depend on $n$, but its size is bounded by a function in $k$. There are examples in the Boolean world, where this dependence on $n$ is used.}. We write $f \leq^{\fpt}_pg$. %\pdc{Do we need $\pfamily$?}
\end{definition}

However p-projection in Valiant's world seems to be {\em weaker} compared to parsimonious poly-time reduction in the Boolean world; therefore we need a stronger notion of reduction for defining algebraic models of the Boolean $\#W$-classes, see \cite{blser_et_al:LIPIcs:2019:11464}. 
%(\pdc{Need to define in Sec 2.2?}). 
That's why we are defining substitutions. We will analogously define it for constant-free model as well.
\begin{definition}[Fpt-substitution]
    \begin{enumerate}
        \item A parameterized family $f = (f_{n,k})$ is an fpt-substitution of another parameterized family $g = (g_{n,k})$ if there are functions $r, s, t,u: \N \to \N$ and polynomials $h_1,\dots,h_{u(r(n)s(k))}\in \K[\X]$ with both $L(h_i)$ and deg$(h_i)$ fpt-bounded such that $r,u$ are $\pbounded$, $s,t$ are computable functions, and $f_{n,k}(\X)=g_{r(n)s(k),k'}(h_1,\dots,h_{u(r(n)s(k))})$ for some $k' \leq t(k)$. We write $f \leq^{\fpt}_s g$. When we allow {\bf unbounded} degree substitution of $h_i$ (i.e.~only $L(h_i)$ is fpt-bounded), we say that $f$ is an $\fptnb$-substitution of $g$. We denote this as $f \leq^{\fptnb}_s g$.
        \item A parameterized family $f = (f_{n,k})$ is a constant-free fpt-substitution of another parameterized family $g = (g_{n,k})$ if there are functions $r, s, t,u: \N \to \N$ and polynomials $h_1,\dots,h_{u(r(n)s(k))}\in \K[\X]$ with both $\tau(h_i)$ and deg$(h_i)$ are fpt-bounded such that $r,u$ are $\pbounded$, $s,t$ are computable and $f_{n,k}(\X)=g_{r(n)s(k),k'}(h_1,\dots,h_{u(r(n)s(k))})$ for some $k' \leq t(k)$. We write $f \leq^{\tau\text{-}\fpt}_s g$. If we remove the degree condition, we get $\fptnb$-substitutions, denoted as $f \leq^{\tau\text{-}\fptnb}_s g$.
    \end{enumerate}
\end{definition}

One can define constant-free fpt-projections analogously. 
The following lemma should be immediate from the definitions, see
\cite{blser_et_al:LIPIcs:2019:11464} for a proof in the case of $\VFPT$.

\begin{lemma} \label{fpt-close}
    $\VFPT, \VFPTnb$ and their constant-free versions ($\VFPT^0$, $\VFPTnb^0$) are {\em closed} under fpt-projections and fpt-substitutions %, fpt-c-reductions 
    (constant-free fpt-projections and constant-free fpt-substitutions, %, constant-free fpt-c-reductions
    respectively). 
    %i.e., using any of the mentioned reduction notions, if $f$ reduces to $g$ and $g\in\VFPT$ ($\VFPTnb, \VFPT^0, \VFPTnb^0$ respectively) then $f\in\VFPT$ ($\VFPTnb, \VFPT^0, \VFPTnb^0$, respectively).
\end{lemma}
\section{Linear counting hierarchy} \label{sec:linear-ch}

In this section, we define the linear counting hierarchy,
a variant of the counting hierarchy, which will allow us to talk
about subexponential complexity.
The original counting hierarchy was defined by Wagner \cite{DBLP:journals/acta/Wagner86}.
We here restrict the witness length to be linear, which is important when dealing
with exponential complexity.
Allender et al.~\cite{DBLP:conf/coco/AllenderKRRV01} also define 
a linear counting hierarchy.
Their definition is not comparable to ours. We use an operator-based definition:
The base class is deterministic polynomial time and the witness length is linearly bounded.
Allender et al.\ use an oracle TM definition: The oracle Turing machine is probabilistic
and linear time bounded, which automatically bounds the query lengths.

\begin{definition}
    Given a complexity class $K$, we define $\cc.K$ to be the class of all languages $A$ such that there is some $B\in K$ and a function $p:\N\rightarrow\N$, $p(n)=O(n^c)$ for some constant $c$, and some polynomial time computable function $f:\{0,1\}^*\rightarrow \N$ such that, \[x\in A\iff|\{y\in\{0,1\}^{p(|x|)}:\langle x,y\rangle\in B\}|>f(x).\]
\end{definition}

We start from $\nchk_0\Pp:=\Pp$ and for all $k\in\N$, $\nchk_{k+1}\Pp:=\cc.\nchk_k\Pp$. Then the {\em counting hierarchy} is defined as $\nch:=\bigcup_{k\geq 0}\nchk_{k}\Pp$.
We now define our linear counting hierarchy:

\begin{definition}
    Given a complexity class $K$, we define $\cc_{\mathsf{lin}}.K$ to be the class of all languages $A$ such that there is some $B\in K$ and a function $\ell:\N\rightarrow\N$, $\ell(n)=O(n)$, and some polynomial time computable function $f:\{0,1\}^*\rightarrow \N$ such that, \[x\in A\iff|\{y\in\{0,1\}^{\ell(|x|)}:\langle x,y\rangle\in B\}|>f(x).\]
\end{definition}

We define $\clink_0\Pp:=\Pp$ and for all $k\in\N$, $\clink_{k+1}\Pp:=\cc_\mathsf{lin}.\clink_k\Pp$. The {\em linear counting hierarchy} 
is $\ch\Pp:=\bigcup_{k\geq 0}\clink_{k}\Pp$.

Now, we slightly modify the above definition to get $\exists_{\mathsf{lin}}.K$ and $\forall_{\mathsf{lin}}.K$ in the following way: $x\in A\iff\exists y\in\{0,1\}^{\ell(|x|)}:\langle x,y\rangle\in B$ and $x\in A\iff\forall y\in\{0,1\}^{\ell(|x|)}:\langle x,y\rangle\in B$, respectively. Clearly, it can be said that $K\subseteq \exists_{\mathsf{lin}}.K\subseteq\cc_{\mathsf{lin}}.K$ and $K\subseteq \forall_{\mathsf{lin}}.K\subseteq\cc_{\mathsf{lin}}.K$.

We can define the linear counting hierarchy in a slightly easier manner.

\begin{definition}\label{def:useful-Ch-lin}
    Given a complexity class $K$, we define $\cc^\prime_{\mathsf{lin}}.K$ to be the class of all languages $A$ such that there is some $B\in K$ and a function $\ell:\N\rightarrow\N$, $\ell(n)=O(n)$, such that \[x\in A\iff|\{y\in\{0,1\}^{\ell(|x|)}:\langle x,y\rangle\in B\}|>2^{\ell(|x|)-1}\]
\end{definition}

It is clear that $\cc^\prime_{\mathsf{lin}}.K\subseteq\cc_{\mathsf{lin}}.K$ for any class $K$. Moreover, by an easy adaption of the proof of \cite[Lemma 3.3]{jtoran}, for any language $K\in\nch$, $\cc_{\mathsf{lin}}.K\subseteq\cc^\prime_{\mathsf{lin}}.K$. Also, from the definition, we can say that $\ch\Pp \subseteq \nch$. Therefore, the following holds.

\begin{fact}\label{fact:ch-rel}
    $\clink_{k+1}\Pp\;=\;\cc^\prime_{\mathsf{lin}}.\clink_k\Pp$.
\end{fact}

We also need a subexponential version of the counting hierarchy.
Let $\SUBEXP = \DTime(2^{o(n)})$.
Then we set $\clink_0\SUBEXP=\SUBEXP$ and for all 
$k\in\N$, $\clink_{k+1}\SUBEXP:=\cc_\mathsf{lin}.\clink_k\SUBEXP$.  
Moreover, $\ch\SUBEXP=\bigcup_{k\geq 0}\clink_k\SUBEXP$.

Here we define a few more terms that we shall use later in Section~\ref{sec5}.
We set $\NP_{\mathsf{lin}} = \exists_{\mathsf{lin}}.\Pp$, $\NP$
with linear witness size. In the same way, we can define 
the levels of the linear polynomial time hierarchy, 
$\Sigma^{\mathsf{lin}}_i$ and $\Pi^{\mathsf{lin}}_i$, by applying the
operators $\exists_{\mathsf{lin}}$ and $\forall_{\mathsf{lin}}$
in an alternating fashion to $\Pp$.
The linear polynomial hierarchy $\PH_{\mathsf{lin}}$ is the
union over all $\Sigma^{\mathsf{lin}}_i$.

From the above definitions, we get the following conclusion.

\begin{fact}
    $\NP_{\mathsf{lin}}\subseteq \PH_{\mathsf{lin}}\subseteq \ch.$
\end{fact}

\section{Connecting Valiant's model to the counting hierarchy}
\label{sec:val-CH}

In this section, we aim to prove that subexponential upper bounds
for exponential sums
imply a collapse of the linear counting hierarchy 
(for a definition, see~\Cref{sec:linear-ch}). 
To show this, we will define a polynomial family $\plsum$ and show that 
$\plsum\in \VFPTnb^0$  is equivalent to 
exponential sums having subexponential circuits (\Cref{thm:plsum-in-vfpt}). 
$\plsum \in \VFPTnb^0$  will imply 
a collapse of the linear counting hierarchy (\Cref{thm:collpase-chsubexp}). 
%We prove these two theorems in the next two subsections.

\subsection{\texorpdfstring{$\log$}{}-variate exponential sum polynomial family}

%\mdc{Why do we call this log-variate. Where is the log in the number of variables?} %\pdc{Isn't $n=\log m^k$ in this setting?} \mdc{OK.}

In this section, we will define a parameterized $\log$-variate exponential sum polynomial family,  
\[
\plsum_{m,k}(g)\;:=\;\sum_{y\in\{0,1\}^{\ell(n)}}\,g_n(\X,y)\;,
\]
where $\X$ has $n$ variables,
$\ell(n) = O(n)$, and $g_n$ has circuits of size $m$ ($n = \Omega(\log m)$), 
and the parameter is $k=\frac{n}{\log m}$. 
$m$ and $k$ are functions of $n$. Note that the running parameter of the family is $m$. When we write $\plsum \in \VFPT$,
we mean that $\{\plsum_{m,k}(g)\}_{m,k} \in \VFPT$ for all families $g$.
We are allowing $g$ to have {\em unbounded} degree, i.e.,~$g$ may not necessarily be a $\pfamily$. We will also be using constant-free circuits computing $g$ in the constant-free context. %\pdc{change to $\ell(n)$}

\subsection{Collapsing of \texorpdfstring{$\ch\SUBEXP$}{}}

The main theorem of the section is the following:

\begin{theorem}\label{thm:collpase-chsubexp}
    If $\plsum\in\VFPTnb^0$, then for every language $L$ in $\ch\SUBEXP$, we have a constant-free algebraic circuit $\chi_L$ so that $x\in L\implies\chi_L(x)=1$, $x\notin L\implies\chi_L(x)=0$ and $\chi_L$ has size $2^{o(n)}$.
\end{theorem}
\begin{proof}
    We prove the above statement by induction on the level of $\ch\SUBEXP$. By definition, $\ch\SUBEXP=\bigcup_{k\geq 0}\clink_k\SUBEXP$. For $k=0$, $\clink_k\SUBEXP=\SUBEXP$. 
    Now by standard arithmetization, we can get a $2^{o(n)}$ size, unbounded degree constant-free circuit for each $L\in\SUBEXP$, so that the above-mentioned condition holds. 
    
    Now, by induction hypothesis say, it is true up to $k$-th level of the hierarchy. We will prove that it is true for the $(k+1)$-th level.
    Take any $B\in\clink_{k+1}\SUBEXP$. By \Cref{fact:ch-rel} and \Cref{def:useful-Ch-lin}, there exists~$A\in\clink_k\SUBEXP$ such that    \[x\in B\;\iff\; |\{y\in\{0,1\}^{\ell(|x|)}\;:\;\langle x,y\rangle\in A\}|\;>\;2^{\ell(|x|)-1}\;,\] where $\ell$ is some linear polynomial. By slight abuse of notation, let $\chi_A$ denote an algebraic circuit capturing the characteristic function for $A$, i.e.,~ 
    \[
       \chi_A(x,y) \;=\;1 \;\iff\; \langle x,y\rangle\in A\;.
    \]
    By the induction hypothesis, we can assume that $\chi_A$
    has size $2^{o(|x|)}$.
    Now, one can equivalently write the following: \[ 
    x\in B\;\iff\;\sum_{y\in\{0,1\}^{\ell(|x|)}}\,\chi_A(x,y)\;>\;2^{\ell(|x|)-1}\;.\]
    In this way, we get an instance of $\plsum$, $\sum_{y\in\{0,1\}^{\ell(|x|)}}\chi_A(x,y)$, where the size of $\chi_A$ is $m=2^{o(|x|)}$ and it computes a polynomial of {\em unbounded degree} (there is no depth-reduction known for Boolean circuits and thus, it cannot be reduced). 
    
    As~$\plsum\in\VFPT^0_{\mathsf{nb}}$, there is an algebraic circuit $C$ such that~$C(x):=\sum_{y\in\{0,1\}^{\ell(|x|)}}\chi_A(x,y)$ and $C$ has 
   subexponential-size by Theorem~\ref{thm:subexp:1}.
    
    Trivially, $\tau(2^{\ell(|x|)-1})\leq \poly(|x|)$. So, we can make $C$ first constant-free and then Boolean by the standard procedure of computing on the binary representation modulo $2^{\ell(n)}$.  
    Let $\Tilde{C}$ is the Boolean circuit that computes the highest bit. We just arithmetize $\Tilde{C}$ and take $\chi_B=\arithmetize(\Tilde{C})$. Each time we convert the arithmetic circuit to a Boolean one and arithmetize the Boolean circuit, we incur only a small polynomial blow-up in size. Therefore, $\chi_B$ has subexponential-size, as desired.
\end{proof}

\begin{remark}
  Clearly, $\ch\Pp\subseteq\ch\SUBEXP$ and hence, $\plsum\in\VFPTnb^0$ implies that every language in $\ch\Pp$ has subexponential-size constant-free algebraic circuits.  
\end{remark}

%\begin{remark}
%     In the definition of $\plsum$, instead of summing over $O(n)$ many variables, we {\em cannot} use $O(n^c)$ many variables. Because if we do so, by the proof idea of \Cref{thm:plsum-in-vfpt}, we can say, the parameter will be $k=n^c/\log m$. And in that case, when we say, $\plsum\in\VFPTnb^0$, this will imply, it has $f(n^c/\log m)\poly(m)$ size constant-free circuit. When we use it in \Cref{thm:collpase-chsubexp}, $m=2^{o(n)}$. Then $n^c/\log m=n^{c^\prime}$ and, as $f$ is an unbounded function, $f(n^{c^\prime})\poly(m)$ can become {\em arbitrarily large}. So, the induction fails.
%\end{remark}
%\pdc{Add another remark for just the algebraic circuit, which we will use in the next section.}

\begin{theorem} \label{thm:subexp:1}
If $\plsum\in\VFPTnb^0$, then $\sum_{y \in \{0,1\}^{\ell(n)}} g(\X, y)$ has
circuits of size $2^{o(n)} \poly(m)$.
\end{theorem}

\begin{proof}
Assume that $\plsum$ has circuits of size $f(n/\log m) \poly(m)$.
We can assume that $f$ is an increasing function. Let 
$i(n) = \max (\{1\} \cup \{j \mid f(j) \le n \})$.
$i(n)$ is nondecreasing and unbounded. Moreover, $f(i(n)) \le n$ for all
but finitely many $n$. 

We will prove that $\sum_{y \in \{0,1\}^{\ell(n)}} g(\X, y)$ has
circuits of size $2^{n/i(n)} \poly(m)$. If $m \ge 2^{n/i(n)}$,
then $f(n/\log m) \le f(i(n)) \le n$, thus there are circuits
of size $n \cdot \poly(m) = \poly(m)$. If $m < 2^{n/i(n)}$,
then let $\hat m = 2^{n/i(n)}$. We can take a circuit $C$ for $g$
and pad it to a circuit $\hat C$ of size $s$ with $\hat m \le s \le O(\hat m)$,
such that $\hat C$ has the same variables as $C$. Then 
let $\hat k = n / \log \hat m$. 
Thus, $\sum_{y \in \{0,1\}^{\ell(n)}} g(\X, y)$ has
circuits of size $f(\hat k) \poly(\hat m) = n \cdot \poly (2^{n/i(n)})$.
\end{proof}

We will need the unbounded version as stated above, but a similar proof also works for the bounded case. The same is true of the
non-constant-free version.
We will also need the following converse direction: 

%\begin{theorem} \label{thm:converse}
%    Say that any family $F_{m,k}(\X)=\displaystyle\sum_{e\in\langle {b(m)\atop k} \rangle}G(\X,e)\in\VWnb^0[\Pp]$ has $2^{o(n)}\poly(m)$ size constant-free circuits where $\tau(G)\leq m$, $n:=k\log m/c$, for some constant $c$ and $b$ is some $\pbounded$ function. Then, $\plsum\in\VFPTnb^0$.
%\end{theorem}

\begin{theorem} \label{thm:converse:alt}
Let $\sum_{y \in \{0,1\}^{\ell(n)}} g(\X,y)$
have circuits of size $2^{o(n)} \poly(m)$ for each $g$ of size $m$. Then
$\plsum \in \VFPTnb^0$.
\end{theorem}

\begin{proof}
Let $C_n$ be a circuit for $\sum_{y \in \{0,1\}^{\ell(n)}} g(\X,y)$ 
of size $2^{O(n / i(n))} \poly(m)$ for some nondecreasing and unbounded function $i$. 
Let $f$ be a nondecreasing function such that $f(i(n)) \ge 2^n$. 
We claim that $\plsum$ has circuits of size $f(k) \poly(m)$
with $k = n/\log m$.
If $m \ge 2^{n/i(n)}$, then $C_n$ has size $\poly(m) \le f(k) \poly(m)$.
Otherwise, $k = n/\log m \ge i(n)$ and therefore $f(k) \ge 2^n$.
Thus, the trivial circuit  for $\sum_{y \in \{0,1\}^{\ell(n)}} g(\X,y)$
has size $f(k) \poly(m)$.
\end{proof} 

%\pdc{Need to write 2 lines and finish theorem 1.}
\section{Integers definable in \texorpdfstring{$\ch\Pp$}{}}\label{sec5}

In \cite[Section 3]{article}, integers are studied that are definable 
in the counting hierarchy. We adapt this notation to the 
\emph{linear} counting hierarchy.
Formally, we are given a sequence of integers $(a(n,k))_{n\in\N,k\leq q(n)}$ for some $\pbounded$ function $q:\N\rightarrow\N$. We can assume that $|a(n,k)|\leq 2^{n^c}$ for some constant $c$. In other words, the bit-size of $a(n,k)$ is at most {\em exponential}, as we think $n,k$ has been represented in binary by $O(\log n)$ bits. Now consider two languages, \begin{align*} 
  \sgn(a)   & :=\{(n,k):a(n,k)\geq 0\}\; \text{ and} \\
  \Bit(|a|) & :=\{(n,k,j,b):j\text{th bit of }|a(n,k)|\text{ is }b\}\;.
\end{align*}
Here in both of these two languages, $n,k,j$ are given in binary representation. 

\begin{definition}
    We say an integer sequence $(a(n,k))_{n\in\N,k\leq q(n)}$ for some $\pbounded$ function $q$ is definable in $\ch\Pp$ whenever both of $\sgn(a)$ and $\Bit(|a|)$ are in $\ch\Pp$.
\end{definition}

{\bf \noindent Chinese remainder language:}~Now, we define another language and make a connection to the definition of an integer sequence to be definable in $\ch\Pp$, via the~{\em Chinese remainder representation}. Given that the bit-size of $a(n,k)$ is at most $n^c$, we consider the set of all primes $p<n^{2c}$. The product of all such primes is $>2^{n^c}$. Therefore, from $a(n,k) \bmod p$, for all primes $p<n^{2c}$, we can recover $a(n,k)$. Consider 
\[
  \CR(a)\;:=\; \left\{(n,k,p,j,b):p\text{ prime},\, p\,<\,n^{2c},\, j\text{-th bit of}~(a(n,k)\,\bmod~p)\text{ is }b\right\}\;.
\]
Now we show an essential criterion for a sequence to be in~$\ch\Pp$.  It is an adaption with some additional modifications and observations from~\cite{HBA}, which were further implemented in~\cite[Theorem 3.5]{article}.

\begin{theorem} \label{thm:defch}
    Let $(a(n,k))_{n\in\N,k\leq q(n)}$ be a  integer sequence of exponential bit-size
     $(|a(n,k)|<2^{n^c})$. Then, $(a(n,k))$ is definable in $\ch\Pp$ iff both $\sgn(a)$ and $\CR(a)$ are in $\ch\Pp$. 
\end{theorem}
\begin{proof}\label{a}
    Our argument goes similar to \cite[Theorem 3.5]{article}, with some further
    modifications.
    
    At first, let us show that for nonnegative sequences, $(a)$ is definable in $\ch\Pp\iff\CR(a)\in\ch\Pp$. To show the $\Rightarrow$ direction, start with a Dlogtime-uniform $\tc^0$ circuit family $(\mathcal{C}_n)_{n}$ which computes the Chinese remainder representation of an $n$-bit number, modulo primes $p<n^2$, from its binary representation. By \cite[Lemma 4.1]{HBA}, $\mathcal{C}_n$ has size $\poly(n)$ and constant depth $D$. Consider the language    
    \[
    L_d\;:=\;\left\{(n,k,F,b): \text{ on input }a(n,k) \text{, gate }F\text{ of }\mathcal{C}_{n^c}\text{ at depth }d\text{ computes bit }b\right\}\;,
    \] 
    for $d\in\{0,\dots,D\}$ and $(n,k,F)$ are given by their binary encoding. \cite[Theorem 3.5]{article} shows that $L_{d+1}\in \cc'.L_d$. But in fact we can say {\em even stronger} that~$L_{d+1}\in \cc_{\mathsf{lin}}'.{L_d}$. This is because when we are given $(n,k,F,b)$ as input by their binary encoding and $F$ is some majority gate at depth $d+1$, we need to check if $(n,k,G,1)$ is in $L_d$ for all gates $G$ at depth $d$ connected to $F$. The lengths
    of the witnesses $(n,k,G,1)$ is $O(\log n)$, which is linear in the input.
    By Dlogtime-uniformity of $(\mathcal{C}_n)_n$, we can check if $G$ is connected to $F$ in polynomial time.  And computing the majority of at most $\poly(n)$ many gates can be done by checking
    \[\left|\{G \mid \text{$G$ connected to $F$ and} (n,k,G,1)\in L_d\}\right|\;>\;2^{\ell(\log n)-1}\;,
    \] 
    for some linear function $\ell$. Hence, our claim is true. The rest of the proof and the other direction is similar to the argument given in~\cite{article}.
    %, using the fact $\clink_{k+1}\Pp=\PP_{\mathsf{lin}}^{\clink_{k}\Pp}$. 
    Hence, $(a)$ is definable in $\ch\Pp\iff\CR(a)\in\ch\Pp$.

    If $(a)$ might have negative entries, then on both sides, we simply add the
    statement ``$\sgn(a) \in \CHlin\Pp$'' (on the left hand side implicitly in the definition of definable).
\end{proof}

Now, we can prove an important {\em closure} property of non-negative integers definable in $\ch\Pp$, which we shall use later.
\begin{theorem}[Closure properties]\label{aaa}
    Let $(a(n,k))_{n\in\N,k\leq q(n)}$ be a non-negative integer sequence for some $\pbounded$ function $q:\N\rightarrow\N$ with $a(n,k)$ having bit-size $<{n^c}$ and it is definable in $\ch\Pp$. Consider the sum and product of $a(n,k)$ defined as follows:
    \[
    b(n)\;:=\;\sum_{k=0}^{q(n)}\,a(n,k)\,\;\;\;\;\text{and}\;\;\;\;c(n)\;:=\;\prod_{k=0}^{q(n)}\,a(n,k)\;.\]
    Then, both of $(b(n))_{n\in\N}$ and $(c(n))_{n\in\N}$ are definable in $\ch\Pp$.
\end{theorem}
\begin{proof}
    The proof is again similar to~\cite[Theorem 3.10]{article}.
    
    {\bf \noindent Part (i): $(b(n))_{n \in \N} \in \ch\Pp$.}\label{p1} 
~By \cite{voll}, we know that iterated addition of $n$ many numbers $0\leq X_1,\dots,X_n\leq 2^n$, given in their binary representation, can be done by Dlogtime-uniform $\tc^0$ circuits. Say this circuit family is $(\mathcal{C}_n)_n$. $\mathcal{C}_n$ has $\poly(n)$ size and constant depth $D$. Now, we can take some $\mathcal{C}_{n^{c^\prime}}$ for some suitable constant $c^\prime$ and using the idea same as in Theorem~\ref{thm:defch}, we can say that $(b(n):=\sum_{k=0}^{q(n)}a(n,k))_{n\in\N}$ is definable in $\ch\Pp$. Note that while we are summing a polynomial number
of numbers, the bit-size for addressing the elements of $a(n,k)$
is $\log n + \log k = O(\log n)$, which is linear in the input size.

\medskip
{\bf \noindent Part (ii): $(c(n))_{n \in \N} \in \ch\Pp$.}
    We first find a generator of $\mathbb{F}_p^{\times}$ for a prime $p$, $p <n^{2c}$.
    The smallest generator $g$ can be characterized by
    \[
      \forall\; 1\leq i<p,\ g^i\neq 1 ~~\text{ and }~~ 
      \forall\; 1\leq \hat g < g, \; \exists \; 1\leq j<p,\ \hat g^j = 1.
    \] 
    The inner checks $g^i\neq 1$ and $\hat g^j = 1$ can be done in polynomial time
    (in $\log n$) by repeated squaring. So checking whether a given $g$ is
    the smallest generator can be done in (the second level of) $\PH_\mathsf{lin} \Pp$.
      
    Also, given $u\in\F_p^{\times}$ and a generator $g$ of $\F_p^{\times}$, finding $1\leq i<p$ so that $u=g^i$ can be done in $\exists_{\mathsf{lin}} \Pp$.
     Note that, 
     \[
      c(n)~\bmod~p\; =  \prod_{k=0}^{q(n)}\,a(n,k)\;\bmod~p\;=\;g^{\sum_{k=0}^{q(n)}\alpha(k,n)}\;.\]
    Finding $g$ and $\alpha(n,k)$ is in $\ch\Pp$, by the above argument. Moreover, the previous part of the proof also shows that $\sum_{k=0}^{q(n)}\alpha(n,k)$ is definable in $\ch\Pp$. Therefore, $(c(n))_{n\in\N}$ is definable in $\ch\Pp$. 
\end{proof}

\begin{corollary}\label{ch}
    Take $a(n,k):=\sigma_{n,k}(1,\dots,n)$, $k\leq n$, where $\sigma_{n,k}(z_1,\dots,z_n)$ is the $k$-th elementary symmetric polynomial on variables $z_1,\dots,z_n$. Then, $(a(n,k))_{n\in\N,k\leq n}$ is definable in $\ch\Pp$.
\end{corollary}
\begin{proof}
    Consider the polynomial 
    \[
    F_n(X)\;:=\;(X+1)\dots(X+n)\;=\;\sum_{k=0}^n\,a(n,k) \cdot X^{n-k}\;.\] Substituting $X$ by $2^{n^2}$, we get that \[d(n)\;:=\;\prod_{j=1}^n(2^{n^2}+j)\;=\;\sum_{k=0}^n a(n,k) \cdot 2^{n^2(n-k)}\;.\]
    And as $a(n,k)<2^{n^2}$, there is {\em no overlap} in the bit representations. Hence, it is enough to show that $(d(n))_{n\in\N}$ is definable in $\ch\Pp$. And by \Cref{aaa}, we only need to prove that $(e(n,k):=2^{n^2}+k)_{n\in\N,k\leq n}$ is definable in $\ch\Pp$, which is indeed true. 
\end{proof}
\section{Connecting the counting hierarchy to the \texorpdfstring{$\tau$}{}-conjecture}
\label{sec:tau}

In this section, we connect the $\tau$-conjecture to the counting hierarchy. Specifically, we show that the collapse of $\ch\Pp$ implies that some explicit polynomial, whose coefficients are definable in $\ch\Pp$, is ``easy''. Formally, we prove the following 
theorem:

\begin{theorem} \label{thm:tau}
    Say, $(a(n))_{n\in\N}$ and $(b(n,k))_{k\leq q(n),n\in\N}$ are both definable in $\ch\Pp$. Here $q$ is some $\pbounded$ function. If $\plsum\in\VFPTnb^0$ then the following holds:
    \begin{enumerate}
        \item $\tau(a(n))=n^{o(1)}$,
        \item If $f_n(X)\;:=\;\sum_{k=1}^{q(n)}b(n,k)X^k$ then $\tau(f_n)=n^{o(1)}$.
    \end{enumerate}
\end{theorem}
\begin{proof}
    We can assume that if $a(n)$ is definable in $\ch\Pp$, $|a(n)|\leq 2^{n^c}$, that is, the bit-size of any integer definable in $\ch\Pp$ is polynomially bounded. Furthermore, 
    if $\plsum \in \VFPTnb^0$, then  every language in $\ch\Pp$ has subexponential-size circuits by~\Cref{thm:collpase-chsubexp}. We will use both facts below. 

    \medskip
    {\bf \noindent Proof of part (1).}~Let $a(n)=\sum_{j=1}^{p(n)}a(n,j)2^j$ be the binary decomposition of $a(n)$ and $p(n)=O(n^c)$. Define a new polynomial: 
    \[
    A_{\lceil\log n\rceil}(Y_1,\dots,Y_{\bit(n)})\;:=\;\sum_{j=0}^{p(n)}a(n,j)Y_1^{j_1}\dots Y_{\bit(n)}^{j_{\bit(n)}}\;,
    \] 
    where $\bit(n)\,:=\,\lceil\log(p(n))\rceil$.
    By our assumption, we can decide if $a(n,j)=b$ by a subexponential-size circuit, given input $n$ and $j$ in binary. Say, $C_r(\Nn,\J)$ is the corresponding circuit, where
    $r = \lfloor \log n \rfloor$. We have
    $C_{r}(n_1,\dots,n_{\lfloor\log n\rfloor+1},j_1,\dots,j_{\bit(n)})=a(n,j)$, where the $n_i$'s and the $j_i$'s are the bits of $n$ and $j$, respectively. Consider the polynomial \[
     F_r(J_1,\dots,J_{cr+1},N_1,\dots,N_{r+1},Y_1,\dots,Y_{cr+1})\;:=\;C_r(\Nn,\J)\cdot \prod_{i=1}^{cr+1}(J_iY_i+1-J_i)\;.
     \]
     Now, by our assumption and Theorem~\ref{thm:collpase-chsubexp}, we can say that $F_r$ has $2^{o(r)}$ 
     size constant-free algebraic circuits (of unbounded degree). Consider the 
     exponential sum
     \[
       \Tilde{F}_r(\Nn,\Y)\;:=\;\sum_{j\in\{0,1\}^{cr+1}}\,F_r(j,\Nn,\Y)\;.
     \] 
     It is an instance of $\plsum$ with $\tau(F_r)=2^{o(r)}$. By assumption, this implies that $\tau(\Tilde{F}_r)=2^{o(r)}$. Finally, note that
    $A_{\lceil\log n\rceil}(\Y)=\Tilde{F}_{r}(n_1,\dots,n_{r+1},\Y)$, and $a(n)=A_{\lceil\log n\rceil}(2^{2^0},\dots,2^{2^{\bit(n)-1}})$. Therefore, \[
    \tau(a(n))\;\leq\; \tau(\Tilde{F}_{r})\;+\;\tau(2^{2^{\bit(n)-1}})\;\leq\; n^{o(1)}\,,
    \]
    as desired.

\medskip
    {\bf \noindent Proof of part (2).}~Again we can assume that $|b(n,k)|$ has polynomially many bits. Let $b(n,k)=\sum_{j=1}^{p(n)}b(n,k,j)2^j$ be the binary decomposition with $p(n)=O(n^{c^\prime})$ and $q(n)=O(n^c)$ . Define 
    \[
    B_{\lceil \log n\rceil}(Y_1,\dots,Y_{\mu(n)},Z_1,\dots,Z_{\lambda(n)})\;:=\;\sum_{k=0}^{q(n)}\sum_{j=0}^{p(n)}b(n,k,j)Y_1^{j_1}\dots Y_{\mu(n)}^{j_{\mu(n)}}Z_1^{k_1}\dots Z_{\lambda(n)}^{k_{\lambda(n)}}\;.
    \]
    Here $\mu(n):=\lceil\log (p(n))\rceil$ and $\lambda(n):=\lceil\log(q(n))\rceil$.
    Let the variable sets be~$\J=(J_1,\dots,J_{c^\prime r+1}), \Nn=(N_1,\dots,N_{r+1}), \Kk=(K_1,\dots,K_{cr+1}), \Y=(Y_1,\dots,Y_{c^\prime r+1}), \mathbf{Z}=(Z_1,\dots,Z_{cr+1})$, where again $r = \lfloor \log n \rfloor$.
    Define a new polynomial $F_r$ as follows:
    \[
    F_r(\J,\Kk,\Nn,\Y,\mathbf{Z})\;:=\; D_r(\Nn,\J,\Kk)\cdot \prod_{m=1}^{c^\prime r+1}(J_mY_m+1-J_m)\prod_{s=1}^{cr+1}(K_sZ_s+1-Z_s)\;
    .\]
    Like in the previous part of the proof, $(D_r(\Nn,\J,\Kk))_r$ is the circuit family for computing $(b(n,k,j))$. In particular, 
    \[D_{r}(n_1,\dots,n_{r+1},j_1,\dots,j_{\mu(n)},k_1,\dots,k_{\lambda(n)})\;=\;b(n,k,j)\;.
    \]
    By our assumption, $D_r$ has $2^{o(r)}$ size constant-free algebraic circuits (of unbounded degree). Consider, 
    \[
    \Tilde{F}_r(\Nn,\Y,\Zz)=\displaystyle\sum_{j\in\{0,1\}^{c^\prime r+1}}\displaystyle\sum_{k\in\{0,1\}^{cr+1}}F_r(j,k,\Nn,\Y,\Zz).
    \] 
    It is an instance of $\plsum$ with $\tau(F_r)$ is $2^{o(r)}$. Since $\plsum\in\VFPTnb^0\implies $ $\tau(\Tilde{F}_r)=2^{o(r)}$. Now, $B_{\lceil\log n\rceil}(\Y,\Zz)=F_{r}(n_1,\dots,n_{r+1},\Y,\Zz)$ and \[f_n(X)\;=\;B_{\lceil\log n\rceil}(2^{2^0},\dots,2^{2^{\mu(n)-1}},X^{2^0},\dots,X^{2^{\lambda(n)-1}})\;.\]
    Therefore, $\tau(f_n)\;\leq\; \tau(B_{\lceil\log n\rceil})\,+\,\tau(2^{2^{\mu(n)}})\,+\,\tau(X^{2^{\lambda(n)}})\;\leq\; n^{o(1)}$, as desired.     
\end{proof}
%\textcolor{red}{Somnath: Toran theorem says $C_kP=PP^{C_{k-1}P}$ which essentially means $\#P=FP\implies PP=P\implies CH=P\implies CH/poly=P/poly$\\
%Check for $C_k-lin=PP-lin^{C_{k-1}-lin}$, if this is true, we might get exp lb on $\#P-lin$}

\begin{theorem}\label{thm:exp-ftp-lb}
If the $\tau$-conjecture is true, then 
$\plsum \notin \VFPTnb$.
\end{theorem}
\begin{proof}
Take the Pochhammer polynomial $p_n(X) = \prod_{i = 1}^{n} (X+i)$. The coefficient of $X^{n-k}$ in $p_n$ will be $\sigma_k(1,\dots,n)$, where $\sigma_k(z_1,\dots,z_n)$ is the $k$-th elementary symmetric polynomial in variables $z_1,\dots,z_n$. And $(\sigma_k(1,\dots,n))_{n\in\N,k\leq n}$ is definable in linear counting hierarchy by \Cref{ch}. By \Cref{thm:tau}, $(p_n)_{n\in\N}$ has $n^{o(1)}$ size constant-free circuit if $\plsum$ is fixed-parameter tractable. But $p_n$ has distinct $n$ many integer roots. So, assuming the 
$\tau$-conjecture, $\plsum$ is not {\em fpt}.
%Therefore, by Theorem~\ref{thm:converse}, $\VWnb[\Pp]$ does not
%have parameterized subexponential algebraic circuits.
\end{proof}

\begin{remark}
    Instead of taking the Pochhammer polynomial, there are many other possible choices
    for some explicit polynomial, see \cite{article}.
\end{remark}

%\pdc{Added this part (below). Please read once.}
Finally, we prove the exponential lower bound for an exponential sum, proving~\Cref{thm:informal-1}.
\begin{theorem}[Exponential algebraic lower bound] \label{thm:formal-exp-lb}
If the $\tau$-conjecture is true, then there exists an~$n$-variate polynomial family $\sum_{y \in \{0,1\}^n} g_n(X,y)$, which requires $2^{\Omega(n)}$-size circuits.
\end{theorem}
\begin{proof}
If the $\tau$-conjecture is true, then \Cref{thm:exp-ftp-lb} shows that $\plsum \notin \VFPTnb$. By the contrapositive statement of \Cref{thm:converse:alt}, the existence of such a hard exponential sum follows.
\end{proof}

\begin{remark}
%\pdc{Is the new remark ok?}
The family $g_n$ simply is a universal circuit of size polynomial in $n$, where the polynomial is large enough to simulate the computation of the Turing machine that shows that the $n$-th Pochhammer polynomial is definable in $\ch\Pp$. 
    %Then by the construction in \Cref{thm:exp-ftp-lb} plus \Cref{thm:collpase-chsubexp},    
    %if this had subexponential size circuits, then the $\tau$-complexity
    %of the Pochhammer polynomial would be subexponential.  
\end{remark}

%\mdc{Say something about question 1?} \pdc{I did, above. Is that okay?}

\section{Preliminaries II: The VW-hierarchy}
\label{sec:VW}

In this section, we define different variants of the $\VW$-hierarchy,
which will be analogous to $\#W$-hierarchy, see \cite{blser_et_al:LIPIcs:2019:11464}. 
We will consider circuits that can have unbounded fanin gates.

\begin{definition}[Weft]
    For an algebraic circuit $C$, the \textit{weft} of $C$ is the maximum number of unbounded fan-in gates on any path from a leaf to the root.
\end{definition}

For $n\geq k\in\N$, let $\langle {n\atop{k}}\rangle$ be the set of all vectors in $\{0,1\}^n$ which have exactly $k$ many $1$s. 

\begin{definition}
    \begin{enumerate}
        \item~A parameterized $\pfamily$ $f_{n,k}(\X)$ is in $\VW[\mathsf{F}]$ iff there exists a $\pbounded$ function $q(n)$ and $\pfamily$ $g_n(\X,y_1,\dots,y_{q(n)})$ such that~$f_{n,k}\leq^{\mathrm{fpt}}_{s}\displaystyle\sum_{\overline{y}\in\langle{{q(n)}\atop k}\rangle}g_n(\X,y_1,\dots,y_{q(n)})$ and $g_n$ can be computed by a polynomial-size 
        formula.
        \item A parameterized family $f_{n,k}(\X)$ is in $\VWnb[\mathsf{F}]$ iff there exists a $\pbounded$ function $q(n)$ and family~$g_n(\X,y_1,\dots,y_{q(n)})$ such that~$f_{n,k}\leq^{\fptnb}_{s}\displaystyle\sum_{\overline{y}\in\langle{{q(n)}\atop k}\rangle}g_n(\X,y_1,\dots,y_{q(n)})$ and $g_n$ can be computed by a polynomial-size formula.
        \item A parameterized $\pfamily$ $f_{n,k}(\X)$ is in $\VW^0[\mathsf{F}]$ iff there exists a $\pbounded$ function $q(n)$ and $\pfamily$ $g_n(\X,y_1,\dots,y_{q(n)})$ such that~$f_{n,k}\leq^{\tau\text{-}\mathrm{fpt}}_{s}\displaystyle\sum_{\overline{y}\in\langle{{q(n)}\atop k}\rangle}g_n(\X,y_1,\dots,y_{q(n)})$ and $g_n$ can be computed by a constant-free, polynomial-size formula.
        \item A parameterized family $f_{n,k}(\X)$ is in $\VWnb^0[\mathsf{F}]$ iff there exists a $\pbounded$ function $q(n)$ and family $g_n(\X,y_1,\dots,y_{q(n)})$ such that~$f_{n,k}\leq^{\tau\text{-}\fptnb}_{s}\displaystyle\sum_{\overline{y}\in\langle{{q(n)}\atop k}\rangle}g_n(\X,y_1,\dots,y_{q(n)})$ and $g_n$ can be computed by a constant-free, polynomial-size formula.
    \end{enumerate}
\end{definition}

In some sense, $\VW[\mathsf{F}]$ is a substitution of a \textit{weighted sum} of formulas. We will define $\VW[\Pp]$ as a weighted sum as above, but summing over an arbitrary circuit of polynomial-size. Similarly, we can define $\VW^0[\Pp]$, and its counterpart in the unbounded setting, i.e.~$\VWnb[\Pp]$, and $\VWnb^0[\Pp]$. 

Finally, we will define the completeness notion:

\begin{definition}
    We will say a parameterized $\pfamily$ $f_{n,k}$ is $\VW[\sfF]$-hard 
    if every $g_{n,k}\in\VW[\sfF]$, $g_{n,k}\leq^{\mathrm{fpt}}_{s} f_{n,q}$.
    Similarly, we can define completeness for $\VW[\Pp]$. 
\end{definition}

We can also define completeness and hardness in the constant-free 
and unbounded models.
\section{Conditional collapsing of VW-hierarchy and applications}
\label{vw-hard}
%\subsection{Conditional collapsing of VW-hierarchy}
%\pdc{I have cleaned up section 3.1 from my side. Reads ok now.}

Let us recall the definition of $k$-degree $n$-variate $(n\geq k)$ {\em elementary symmetric polynomial} % $S_{n,k}$: 
$
\sigma_{n,k}(\X)\;:=\;\sum_{y\in\langle {n\atop{k}}\rangle}\,X_1^{y_1}X_2^{y_2}\dots X_n^{y_n}\;.
$
It is known that $(\sigma_{n,k})_n \in \VP^0$, with a simple dynamic programming algorithm; see~\cite[Section 4]{jukna2016optimality}. 
Let us define a new polynomial family $B_{n,k}(\X)$, which will be important in the latter part of the section:
%\[
$B_{n,k}(\X)\;:=\;\sum_{t=0}^{n-k}\,(-1)^t\binom{k+t}{k} \cdot \sigma_{n,k+t}(\X)\;$.
%\]
The following claim is crucial:
\begin{claim}\label{cl:indicator-poly}
For $y\in\{0,1\}^n$, $B_{n,k}(y)= \begin{cases}
    1 & \text{ if }y \in \langle{n\atop k}\rangle.\\
    0 & \text{ otherwise}.
\end{cases}
$
\end{claim} 
\begin{proof}
    For a string $y\in\{0,1\}^n$, we will call the {\em weight} of $y$, denoted $\wt(y)$, the number of $1$'s present in $y$. Note that if $\wt(y) < k$, then $\sigma_{n,k}(y)=0$ implying $B_{n,k}(y) = 0$.
 Similarly if $\wt(y)=k$, then $B_{n,k}(y) = \sigma_{n,k}(y)$, which will be exactly equal to $1$.
Now if $\wt(y) = k+r$ where $r>0$, then 
\begin{align*}
    B_{n,k}(y)\;=\;\sum_{t=0}^{n-k}\,(-1)^t\binom{k+t}{k} \cdot \sigma_{n,k+t}(y)\;&=\;\sum_{t=0}^{r}\,(-1)^t\binom{k+t}{k} \cdot \sigma_{n,k+t}(y)\\
    \;&=\;\sum_{t=0}^{r}\,(-1)^t\binom{k+t}{k} \cdot \binom{k+r}{k+t} \\
    \;&=\;\sum_{t=0}^{r}\,(-1)^t\frac{(k+r)!}{k!t!(r-t)!}\;.
\end{align*}
Let us further define the tri-variate polynomial $Q(x,y,z) := (x+y-z)^{k+r} \in \Z[x,y,z]$. Note that the coefficient of $x^k$ in $Q(x,y,z)$ is \begin{align*}
    \sum_{t=0}^r\,y^{r-t}z^t(-1)^t \cdot \frac{(k+r)!}{k!t!(r-t)!}\;.
\end{align*}
Now putting $y = z = 1$, we get the coefficient exactly equal to $B_{n,k}(y)$; since $r\neq0$, we can say that the coefficient of $x^k$ in $Q(x,1,1)$ is $0$, which finally implies that $B_{n,k}(y) = 0$.
\end{proof}

Now we are ready to prove the following transfer theorem from the parameterized Valiant's classes to Valiant's algebraic models. \begin{theorem}\label{vw-collapse}
    $\VW^0[\Pp]\;\neq\; \VFPT^0\;\implies\; \VP^0\;\neq\; \VNP^0$. Similarly, $\VW[\Pp]\;\neq\; \VFPT\;\implies\;\VP\;\neq\;\VNP$. %\pdc{ask}
\end{theorem}

\begin{proof}
We will prove the contraposition. Assume that $\VP^0 = \VNP^0$. As mentioned before, we know that $(\sigma_{n,k})_n \in \VP^0$. Further, since $k \in [n]$, for $t \le n-k$, it is trivial to see that $\tau(\binom{k+t}{k}) = n^{O(1)}$. Therefore, for each $ 0 \le t \le n-k$, $(-1)^t \binom{k+t}{k} \cdot \sigma_{n,k+t}(\X)$ has a $\VP^0$-circuit. Since $\VP^0$ is closed under polynomially many additions, it follows that $(B_{n,k})_n \in \VP^0$.

%\pdc{Edited till this point above. @Somnath modif y the below proof as we discussed. I will edit after that. }

Let $q_{n,k} \in \VW^0[\Pp]$. By definition, there is a polynomial family $p_{n,k}$ of the above form $p_{n,k}(\X)\,:=\,\sum_{y\in\langle{n\atop k}\rangle}\,g_n(\X,y)$, where $g_n(\X,\Y)$ is in $\VP^0$, such that~$q_{n,k}\leq^{{fpt}}_s p_{n,k}$. By Claim~\ref{cl:indicator-poly}, it follows that \[
p_{n,k}\;=\;\sum_{y\in\{0,1\}^n}\,g_n(\X,y) \cdot B_{n,k}(y)\;.
\]
We have already proved above that $B_{n,k}$ has $\poly(n)$ sized constant-free circuits. Hence, $g_n(\X,y) B_{n,k}(y)$ has constant-free $\poly(n)$-size circuit. Therefore, by definition and our primary assumption, it follows that~$p_{n,k}\in\VNP^0 = \VP^0 \subseteq \VFPT^0$.
Since, $\VFPT^0$ is {\em closed} under constant-free fpt-substitution (\Cref{fpt-close}), it follows that $q_{n,k}\in\VFPT^0$, implying $\VW^0[\Pp] \subseteq \VFPT^0$.

    The proof in the usual (not constant-free) model also follows essentially along the same line as above.
\end{proof}

\begin{remark}\label{remark:vfpt-to-vp-lb}
The above theorem holds in the unbounded regime as well, i.e.,~$\VWnb^0[\Pp]\;\neq\; \VFPTnb^0\;\implies\; \VPnb^0\;\neq\; \VNPnb^0$ (which further implies $\VP^0 \ne \VNP^0$, see \cite{koiran2011interpolation}). Similarly, $\VWnb[\Pp]\;\neq\; \VFPTnb\;\implies\;\VPnb\;\neq\;\VNPnb$.  
\end{remark}

We now aim to prove a {\em conditional separation} of $\VWnb^0[\Pp]$ and $\VFPTnb^0$, 
by showing that $\VWnb^0[\Pp]=\VFPTnb^0$ implies a collapse of the linear counting hierarchy. 
To show this, we will show that  $\VWnb^0[\Pp]=\VFPTnb^0\implies\plsum\in \VFPTnb^0$ (\Cref{thm:plsum-in-vfpt}) from which the collapse of
the linear counting hierarchy follows.

%\pdc{What is the following theorem? Seems abrupt. needs a line to connect.}

\begin{theorem}\label{plsumthm}
    Let $f(\X)=\displaystyle\sum_{y\in\{0,1\}^{\ell(n)}}g(\X,y)$, where $\ell(\cdot)$ is a linear function and $g$ is computed by an arithmetic circuit of size $m = 2^{O(n^c)}$ for some constant $c$. %=2^{n^{\Omega(1)}}$. 
    Then, $f(\X)$ can be written as 
    \[f(\X)\;=\;\sum_{e\in\langle {b(m)\atop k} \rangle}\,G(\X,e)\;,\] for some $\pbounded$ function $b$ and $k=\ell(n)/\log m$ and $G$ has $\poly(m)$ size circuits.
\end{theorem}
\begin{proof}
    Let $f(\X)$ be an instance of $\plsum$, i.e.,~$f(\X)=\sum_{y\in\{0,1\}^{n}} g(\X,y)$, where $g(\X,\Y)$ has size $m$ %=2^{n^{\Omega(1)}}$ 
    constant-free circuit. Here we mention that, although we just take sum over {\em $n$ variables} here for the ease of presentation, the same proof also works if we sum over $\ell(n)$ many variables for some linear function $\ell$.
    
    Let us partition the variable set $\Y=\{Y_1,\dots,Y_n\} = E_1\sqcup\dots\sqcup E_k$. Here $k=n/\log m$, and for all $i$, $|E_i|=\log m$. For each $S\subseteq E_i$, we take a {\em new variable} $Z_i^S$ and we do this for all $i$. Define $\overline{Z_i} :=\{Z_i^S \,:\, S \subseteq E_i\}$ and $\bZ=\bigcup_i\overline{Z_i}$. The number of $\bZ$-variables is $2^{\log m} \cdot k$, which
    is polynomial in $m$. 
    
    Let us call an assignment of $\bZ$ variables a {\em good assignment}, if {\em exactly} one variable in each set $\overline{Z_i}$ is set to be $1$.  Below we show that there is a one-to-one correspondence between
    $\{0,1\}$ assignments to the $\Y$ variables 
    and {\em good assignments} to the $\bZ$ variables.
    
    Let $\varphi$ be a homomorphism from $R[\Y] \to R[\bZ]$, where $R:=\F[\X]$, such that $\varphi: Y_i \mapsto \prod_{S\subseteq E_i,\ Y_j\not\in S}(1-Z_i^S)$. Let us define $\Tilde{g}(\X,\bZ) := \varphi(g)$. Now let us fix an assignment $y \in \{0,1\}^n$ to the $\Y$ variables. 
    We construct a corresponding good assignment of $\bZ$.
    For each $E_i$ of $\Y$, we have some $S_i\subseteq E_i$ such that {\em each} variable of $E_i$, which is in $S_i$, gets value~$1$. The remaining variables in $E_i\setminus S_i$ get value $0$ (so that it corresponds to $y$). Pick this particular $S_i\subseteq E_i$. Note that this $S_i$ is {\em unique} (it can be the empty set). Now set $Z_i^{S_i}=1$ and $Z_i^{S}=0$, if $S\neq S_i$, for all $i\in[k]$.
    
    Each variable in $\bigcup_{i} S_i$ gets the value $1$ and variables in $\bigcup_i(E_i\setminus S_i)$ are assigned $0$. Under the map $\varphi$, any $Y_j\in E_1\setminus S_1$ is replaced by $\prod_{S\subseteq E_1,\ Y_j\not\in S}(1-Z_1^S)$. Since, $S_1\subseteq E_1$ and $Y_j\notin S_1$, $(1-Z_1^{S_1})$ occurs in the product. And, hence the product becomes $0$. Now, let $Y_{\ell}\in S_1$ and $\varphi(Y_{\ell})=\prod_{S\subseteq E_1,\ Y_{\ell}\not\in S}(1-Z_1^S)$. As $Y_{\ell}\in S_1$, $(1-Z_1^{S_1})$ does not contribute to the product. Thus, under the assignment defined before, $\varphi(Y_{\ell})$ becomes $1$. This argument holds for any $E_i$. Therefore, one can conclude that
    \[f\;=\;\sum_{e:\  e\text{ is a }good\text{ assignment }}\Tilde{g}(\X,e)\;.\]

    Note that the weft of the circuit for $\Tilde{g}$ has increased by $1$ (from that of $g$), and the size has also increased by a polynomial (in $m$) factor. To capture a $k$-weight good assignment exactly, define a new polynomial $ p(\bZ) \in \F[\bZ]$ as follows: 
    \[p(\bZ)\;:=\;\prod_{i=1}^k\,\bigg(\sum_{S\subseteq E_i} Z^{S}_j\bigg)\;.
    \]
    Clearly, $p$ has a weft-$2$ circuit of size $\poly(m)$. Further, it is simple to see that for any $k$-weight $\{0,1\}$ assignment $e$ to the $\bZ$ variables, $p(e)=1$ iff $e$ is a $good$ assignment because from each of the product terms, only one variable will survive. Therefore, \[f\;=\;\sum_{e\in\langle {b(m)\atop k} \rangle}\;p(e)\cdot \Tilde{g}(\X,e)\;,\;\;\;\;\;\;\;\text{where}\;\;b(m)=|\bZ|\;.\]
    We set $G(\X,\bZ):=p(\bZ)\Tilde{g}(\X,\bZ)$. By the construction, $\Tilde{g}$ has weft $\leq t+1$, $p$ has weft $\leq 2$, and $\Tilde{g},p$ have $\poly(m)$ size circuits. So, this ends our proof. 
\end{proof}

\begin{remark}
   The construction above increases the weft by one. 
\end{remark}

\begin{corollary}\label{thm:plsum-in-vfpt}
$\VWnb^0[\Pp]=\VFPTnb^0\implies\plsum\in \VFPTnb^0$.
\end{corollary}
\begin{proof}
    In \Cref{plsumthm} we have reduced an instance of $\plsum$ to an instance of $\VWnb^0[\Pp]$ with parameter $k=\ell(n)/\log m$. By our assumption $\VWnb^0[\Pp]=\VFPTnb^0$ and thus we can say that $\plsum\in\VFPTnb^0$. 
    % with the parameter $\hat{k}=n/\log m$.
\end{proof}

\begin{remark}
If one restricts $\plsum$ to exponential sums over $g$, where $g$ is a $\pfamily$ (i.e.,~it has polynomial degree and size), denoted $\plsum_{bd}$ (bd for bounded-degree), then the above proof similarly implies that $\VW^0[\Pp]=\VFPT^0\implies\plsum_{bd} \in \VFPT^0$.
\end{remark}

%\pdc{The below part is copy-pasted from the old icalp file, since this shows the fpt lb.}
Similarly, we also prove a lower bound for the class $\VWnb[\Pp]$, assuming an fpt lower bound on $\plsum$. %A similar result hold for the bounded case. For a proof, see~\Cref{sec:appendix-sec4}.

\begin{theorem} \label{thm:converse}
    Say that any family $F_{m,k}(\X)=\displaystyle\sum_{e\in\langle {b(m)\atop k} \rangle}G(\X,e)\in\VWnb^0[\Pp]$ has $2^{o(n)}\poly(m)$ size constant-free circuits where $\tau(G)\leq m$, $n:=k\log m/c$, for some constant $c$ and $b$ is some $\pbounded$ function. Then, $\plsum\in\VFPTnb^0$.
\end{theorem}
\begin{proof}
    Take an instance of $\plsum$, $f(\X)=\sum_{y\in\{0,1\}^{\ell(n)}}g(\X,y)$, for some $\ell(n)=O(n)$. And $g$ has a constant-free circuit of size $m$. By \Cref{plsumthm}, we can make it an instance of $\VW^0[\Pp]$ and say, \[f=\sum_{e\in\langle {b(m)\atop k} \rangle}\Tilde{g}(\X,e)\;,\;\;\;\;\;\;\;\text{where}\;\;b\;\text{ is }\pbounded,\;\;k=\ell(n)/\log m\]
    By our assumption, $f$ has a constant-free circuit of size $2^{o(n)}\poly(m)=2^{O(n/i(n))}\poly(m)$ for some unbounded and non-decreasing function $i:\N\rightarrow\N$. Let $h$ be a non-decreasing function, so that $h(i(n))\geq 2^n$. We shall prove that $f$ has $h(k)\poly(m)$ size constant-free circuit. If $m\geq 2^{n/i(n)}$, clearly, $f$ has $\poly(m)$ size constant-free circuit. Otherwise, if $m<2^{n/i(n)}$, this will imply $i(n)\leq n/\log m=k$. And hence, $h(k)\geq 2^n$. So, $f$ has $h(k)\poly(m)$ size constant-free circuit.
\end{proof}

\section{Restricted permanent}
\label{sec:complete}

A \emph{cycle cover} of a directed graph is a collection of node-disjoint directed cycles
such that each node is contained in exactly one cycle. Cycle covers
of a directed graph stand in one-to-one relation with permutations of the nodes.

\begin{definition}
A cycle cover is $(k,c)$-restricted, if it contains one cycle of length $k$
and all other cycles have length $\le c$.
\end{definition}

Let $G = (V,E)$ be directed graph and $w: E \to R$ be a weight function.
Here $R$ is a ring and typically the ring of polynomials.
The weight of a cycle cover $C$ of $G$ is the product of the weights
of the edges in it, that is, $ w(C) = \prod_{e \in C} w(e)$.

\begin{definition}
The $(k,c)$-restricted permanent of an edge-weighted directed graph $G$
is 
\[
  \per^{(k,\le c)}(G) = \sum_{C} w(C), 
\] 
where the sum is over all $(k,c)$-restricted cycle covers.
\end{definition}

If $X = (X_{i,j})$ is a variable matrix, then $\per_{n}(X)$
is the permanent of the complete directed graph with the
edge weights $w(i,j) = X_{i,j}$. The $(k,c)$-restricted
permanent family $\per^{(k,\le c)} = (\per_{n}^{(k,\le c)}(X_n))$,
where $X_n$ is an $n \times n$-variables matrix. 
$\per^{(k,\le c)}$ is a parameterized family, $n$ is the input size,
$k$ is the parameter, and $c$ will be some constant to be determined later.

On general graphs, the restricted permanent is very powerful,
even if we keep the parameter fixed.

\begin{proposition} \label{prop:22perm}
The $(2,2)$-restricted permanent family is $\VNP$-complete.    
\end{proposition}
\begin{proof}
We reduce from the matching polynomial on undirected graphs.
Given a matching $M$ of the complete undirected graph, we can map it to a 
$(2,2)$-restricted cycle cover $C$ of the complete directed graph (with self-loops),
by mapping each edge $\{i,j\} \in M$ to the $2$-cycle $(i,j),(j,i)$. Nodes
that are not covered by $M$ are covered by self-loops in $C$.
This is a one-to-one correspondence. Therefore, if we substitute
$X_{i,i} = 1$, $1 \le i \le n$ and $X_{i,j} = 1$ for $i > j$, then
we get the matching polynomial out of $\per_n^{(2,\le 2)}(X)$.
\end{proof}

%\begin{definition}
%The girth of an undirected graph is the length of the shortest cycles in the graph.
%\end{definition}

If we restrict the underlying graph appropriately, then the restricted
permanent is complete for the class $\VW[\sfF]$.
Recall that the girth of an undirected graph is the length of a shortest cycle in the graph.
When we talk of the girth of a directed graph, we mean the girth of the graph
when we disregard the direction of edges. 
Furthermore, when we talk about the treewidth of a directed graph, we mean
the treewidth of the underlying undirected graph.
%(For the reader's convenience, we recall the definition of treewidth 
%in the appendix, see \Cref{def:treewidth}.)

\begin{definition} \label{def:cbnice}
A directed graph $G = (V,E)$ is \emph{$(c,b)$-nice} if we can partition the nodes $V = V_1 \cup V_2$
into two disjoint sets, such that
\begin{enumerate}
\item the graph induced by $V_1$ has girth $> c$ (not counting self-loops),
\item every node in $V_1$ has a self-loop, and
\item the graph induced by $V_2$ has tree-width bounded by $b$.
\item every cycle that contains vertices from $V_1$ and $V_2$ has length $> c$.
\end{enumerate}
\end{definition}

Our main result is the following completeness result.
%The proofs are rather long and can be found in 
%Sections~\ref{sec:hardness} and~\ref{sec:upper}.

\begin{theorem} \label{thm:kcper:easy}
Let $c$ and $b$ be constants.
    Let $(G_n)$ be a family of $(c,b)$-nice graphs. 
    Then the $(k,c)$-restricted permanent is in $\VW[\sfF]$.
\end{theorem}

\begin{theorem} \label{cor:kcper:hard}
Let the underlying field have characteristic $0$.
There is a constant $b$ and 
a family of $(4,b)$-nice graphs $(H_n)$ such that the $(3k,4)$-restricted
permanent of $H_n$ forms a family of $\VW[\sfF]$-hard polynomials.
\end{theorem}
\section{Hardness of the restricted permanent}
\label{sec:hardness}
For the reader's convenience, we recall the definition of tree-width:

\begin{definition} \label{def:treewidth}
\begin{enumerate}
\item A \emph{tree decomposition} of an undirected  graph $G = (V, E)$ is a pair 
$(\{X_i \mid i \in I\}, T = (I, F))$, 
where $\{X_i \mid i \in I\}$ is a collection of subsets of $V$ and $T = (I, F)$
is a tree such that:
\begin{itemize}
\item  $\bigcup_{i \in I} X_i = V$.
\item For all $\{v, w\} \in E$, there exists an $i \in I$ with $v, w \in X_i$.
\item For every $v \in V$, $T_v = \{i \in I \mid v \in X_i\}$ is connected in $T$.
\end{itemize}
\item The \emph{width} of a tree decomposition is $\max_{i \in I} |X_i| - 1$. 
The \emph{treewidth} of a graph $G$ is the minimum width over all the tree
decompositions of $G$.
\end{enumerate}
\end{definition}

The $X_i$ are also called \emph{bags}. The treewidth of a directed graph
is the treewidth of the underlying undirected graph.

Now, we are given a formula $F$ in variables $X_1,\dots,X_m$ and $Y_1,\dots,Y_n$.
We call the polynomial computed by $F$ also $F(X_1,\dots,X_m,Y_1,\dots,Y_n)$.
We are interested in the polynomial
\[
  P(X_1,\dots,X_m) = \sum_{e \in \ones{n}{k}} F(X,e)
\]
We assume that the formula is layered, that is, along each path the addition and multiplication gates alternate. The top gate is an addition gate and each input
gate is fed into a multiplication gate.

Our goal is to write $P$ as an fpt-projection of a $(k,c)$-restricted
permanent on a $(c,b)$-nice graph $H$ for certain constants $c$ and $b$.
The construction will have two main components. One
corresponds to the formula $F$, the other one is similar to the rosetta graph in Valiant's proof of the $\sharpP$-hardness of the permanent,
see e.g.\ \cite{Buergisser:00}.
The first component will be the bounded treewidth part of $H$, 
the second one will be the high girth part.

\subsection{The graph $G_1$}
\label{sec:G1}

We first design a graph $G_1$ from $F$. $G_1$ will have \emph{iff-coupled} edges (pairs of edges).
When we have a pair of iff-coupled edges, then either both of them appear in
a cycle cover or none of them, see also \cite{Buergisser:00}. 
Later, we will enforce this by connecting
iff-coupled edges with appropriate gadgets.

A cycle cover of $G_1$ is \emph{consistent} if it contains either both edges of such a pair or none. $G_1$ will have the property that 
\begin{itemize}
    \item the sum of the weights of all consistent cycle covers in $G_1$ is $F$
    \item and each cycle cover has cycles of length at most two.
\end{itemize}
The graph will be constructed in an iterative manner, adding new nodes and edges gradually. Each parse tree of $F$ will correspond to one consistent cycle cover and vice versa.
Recall that a \emph{parse tree} of an algebraic formula is a subtree that contains the root,
for each multiplication gate it contains all children and for each addition gate
it contains exactly one child, see \cite{DBLP:journals/jc/MalodP08}.

\subsubsection{Input gates}
Input gates are realized as depicted in Figure~\ref{fig:input}. Assume that $L$ is the label
of the input gate.
The gate has the following properties:
\begin{itemize}
    \item If the top node is externally covered (meaning that the
    gate is in the parse tree), then there is exactly one
    consistent cycle cover with weight $L$ (middle, drawn in blue).
    \item If the top node is uncovered, then there is exactly one consistent
    cycle cover with weight $1$ (right-hand side, drawn in blue).
\end{itemize}

\begin{figure}
\begin{minipage}[b]{0.45\textwidth}
\centering
\begin{tikzpicture}[scale=0.75,->]
    \node[vertex] (t) at (1,2) {};
    \node[vertex] (b) at (1,0) {};
    \path (t) edge[bend left=20] node [right] {} (b);
    \path (b) edge[bend left=20] node [right] {} (t);
    \path (b) edge[loop below] node [below] {$L$} (b);
    \node[vertex] (t1) at (4,2) {};
    \node[vertex] (b1) at (4,0) {};
    \path (t1) edge[bend left=20] node [right] {} (b1);
    \path (b1) edge[bend left=20] node [right] {} (t1);
    \path (b1) edge[loop below, blue] node [below] {$L$} (b1);
    \node[vertex] (t2) at (6,2) {};
    \node[vertex] (b2) at (6,0) {};
    \path (t2) edge[bend left=20, blue] node [right] {} (b2);
    \path (b2) edge[bend left=20, blue] node [right] {} (t2);
    \path (b2) edge[loop below] node [below] {$L$} (b2);
\end{tikzpicture}
\caption{The input gadget and the two ways how to cover it (drawn blue). \label{fig:input}}
\end{minipage}\hspace*{0.1\textwidth}\begin{minipage}[b]{0.45\textwidth}
\centering
\begin{tikzpicture}[scale=0.75,->]
    \node[vertex] (t) at (1,2) {};
    \node[vertex] (b) at (1,0) {};
    \path (t) edge[bend left=20] node [right] {} (b);
    \path (b) edge[bend left=20] node [midway,right] (m0) {} (t);
    \path (t) edge[loop above] node [below] {} (b);
    \node[vertex] (t1) at (3,2) {};
    \node[vertex] (b1) at (3,0) {};
    \path (t1) edge[bend left=20] node [midway, left] (m1) {} (b1);
    \path (b1) edge[bend left=20] node [midway, right] (m2) {} (t1);
    \path (t1) edge[loop above] node [below] {} (b1);
    \node[vertex] (t2) at (5,2) {};
    \node[vertex] (b2) at (5,0) {};
    \path (t2) edge[bend left=20] node [left, midway] (m3) {} (b2);
    \path (b2) edge[bend left=20] node [right, midway] (m4) {} (t2);
    \path (t2) edge[loop above] node [below] {} (b2);
    \path [draw, <->, dashed, gray, shorten <= 3pt, shorten >= 3pt] (m0) -- (m1);
    \path [draw, <->, dashed, gray, shorten <= 3pt, shorten >= 3pt] (m2) -- (m3);
    \node (m5) at (7,1) {\dots};
    \path [draw, <->, dashed, gray, shorten <= 3pt, shorten >= 3pt] (m4) -- (m5);
\end{tikzpicture}
\caption{The multiplication gadget. Iff-couplings are drawn as dashed
bidirected edges.\label{fig:mult}}
\end{minipage}
\end{figure}

\subsubsection{Multiplication gates} 

Multiplication gates are realized as depicted in Figure~\ref{fig:mult}.
For each child of the multiplication gate, we have one 2-cycle.
These 2-cycles are iff-coupled. The bottom node of each 2-cycle
will be the top node of an input gate or the yet-to-define addition gate.

The gate has the following properties:
\begin{itemize}
\item If the left-most edge is in a consistent cycle cover, then this 
consistent cycle cover contains all 2-cycles of the gadget. 
(This means that the multiplication gate is in the parse tree.)
\item If the left-most edge is not in a consistent cycle cover,
then all two nodes will be covered by self-loops. The bottom
nodes have to be covered externally.
(This means that the multiplication gate is not in the parse tree.)
\end{itemize}

The bottom nodes of each 2-cycle will be the top-node of the input gates
or (yet to be defined) addition gates that are fed into the multiplication gate.

\begin{itemize}
\item If the multiplication gate is in the parse tree, then all its children 
are in the parse tree. In this case, the bottom nodes
of the 2-cycles are covered within the multiplication gadget.
These bottom nodes are the top nodes of the input gadgets and addition gadgets.
For these gates, their top nodes are now externally covered, which means
that the corresponding gates are in the parse tree, as it should be.
\item If the multiplication gate is not in the parse tree,
then all its children are not in the parse tree. In this case, the bottom
nodes are not covered within the multiplication gadget. Hence, they
need to be covered in the input gadgets or additions gates, which means that
the corresponding gates are not in the parse tree, too, see the following subsections.
\end{itemize}

\subsubsection{Addition gates}
The addition looks as drawn in Figure~\ref{fig:add}. 
It has a 2-cycle at the top and then has one 2-cycle for each child.
It has the following properties:
\begin{itemize}
    \item If the top node is not covered externally (that is, the addition gate is not in
    the parse tree), then there is exactly one consistent cycle cover.
    \item If the top node is covered externally (that is, the gate is in the parse tree),
    then there are $t$ different cycle covers, where $t$ is the number of children,
    one for each 2-cycle in the bottom row.
    This reflects the fact that in a parse tree, an addition gate has exactly one child.
\end{itemize}

The Figure~\ref{fig:add:cover} shows the situation when the top node
is not covered externally on the left-hand side. On the right-hand side, it shows the situation
in the second case. Here are $t$ covers,
each of them contains one 2-cycle and $t-1$ self-loops.

\begin{figure}
\begin{minipage}{0.3\textwidth}
\centering
\begin{tikzpicture}[scale=0.75,->]
    \node[vertex] (t) at (3,2) {};
    \node[vertex] (b) at (3,0) {};
    \path (t) edge[bend left=20] node [right] {} (b);
    \path (b) edge[bend left=20] node [midway,right] (m0) {} (t);
    \node[vertex] (c1) at (1,-2) {};
    \node[vertex] (c2) at (3,-2) {};
    \node[vertex] (c3) at (5,-2) {};
    \path (c1) edge[loop below] (c1) edge[bend left=20] (b);
    \path (b) edge[bend left=20] (c1);
    \path (c2) edge[loop below] (c2) edge[bend left=20] (b);
    \path (b) edge[bend left=20] (c2);
    \path (c3) edge[loop below] (c3) edge[bend left=20] (b);
    \path (b) edge[bend left=20] (c3);
    \node at (5.3,-1) {\dots};
\end{tikzpicture}    
\caption{The addition gadget. If the corresponding gate has fanin $t$, then there
are $t$ nodes at the bottom. \label{fig:add} }
\end{minipage}\hspace*{0.1\textwidth}\begin{minipage}{0.6\textwidth}
\centering
\begin{tikzpicture}[scale=0.75,->]
    \node[vertex] (t) at (3,2) {};
    \node[vertex] (b) at (3,0) {};
    \path (t) edge[bend left=20, blue] node [right] {} (b);
    \path (b) edge[bend left=20, blue] node [midway,right] (m0) {} (t);
    \node[vertex] (c1) at (1,-2) {};
    \node[vertex] (c2) at (3,-2) {};
    \node[vertex] (c3) at (5,-2) {};
    \path (c1) edge[loop below, blue] (c1) edge[bend left=20] (b);
    \path (b) edge[bend left=20] (c1);
    \path (c2) edge[loop below, blue] (c2) edge[bend left=20] (b);
    \path (b) edge[bend left=20] (c2);
    \path (c3) edge[loop below, blue] (c3) edge[bend left=20] (b);
    \path (b) edge[bend left=20] (c3);
    \node at (5.3,-1) {\dots};
\end{tikzpicture}~~~~~~~~~\begin{tikzpicture}[scale=0.75,->]
    \node[vertex] (t) at (3,2) {};
    \node[vertex] (b) at (3,0) {};
    \path (t) edge[bend left=20] node [right] {} (b);
    \path (b) edge[bend left=20] node [midway,right] (m0) {} (t);
    \node[vertex] (c1) at (1,-2) {};
    \node[vertex] (c2) at (3,-2) {};
    \node[vertex] (c3) at (5,-2) {};
    \path (c1) edge[loop below] (c1) edge[bend left=20, blue] (b);
    \path (b) edge[bend left=20, blue] (c1);
    \path (c2) edge[loop below, blue] (c2) edge[bend left=20] (b);
    \path (b) edge[bend left=20] (c2);
    \path (c3) edge[loop below, blue] (c3) edge[bend left=20] (b);
    \path (b) edge[bend left=20] (c3);
    \node at (5.3,-1) {\dots};
\end{tikzpicture}
\caption{Lefthand side: The covering (drawn blue) if the top node is not externally covered.
Righthand side: The covering if the top node is externally covered. One of the bottome nodes
is covered by a $2$-cycle. This is the child in the corresponding parse tree.
\label{fig:add:cover}}
\end{minipage}
\end{figure}

The children of an addition gate are all multiplication gates.
The left-most edge of the multiplication gate will be iff-coupled to
one of the edges of the corresponding 2-cycle in the bottom row.

\subsubsection{Putting it all together} 
\label{sec:G1:putting}

Let $F$ be the given formula. We construct the
corresponding graph $G_F$ recursively:
\begin{itemize}
    \item If $F$ consists of one node (an input node), then $G_F$
    is the corresponding input gadget.
    \item If the top gate of $F$ is a multiplication gate,
    then let $F_1,\dots,F_t$ be its children (summation gates).
    We take the graphs $G_{F_1},\dots,G_{F_t}$ and identify their top nodes
    with the bottom nodes of the corresponding 2-cycles in the multiplication gadget
    to get $G_F$.
    \item If the top gate of $F$ is an addition gate with children $F_1,\dots,F_t$,
    then we take the corresponding graphs $G_{F_1},\dots,G_{F_t}$
    and iff-couple the left of the left-most 2-cycle of the top addition gate
    to one of the edges of the corresponding 2-cycle of the addition gate.
%    \item Input gadgets labelled with $Y_i$ get the label $1$ instead.
\end{itemize}

The graph $G_1$ will now be the graph $G_F$ with one 2-cycle 
attached to the top node, when the top node is an addition gate. 
This ensures that the top node of the addition
gadget is always externally covered, so the addition gate is always in
a parse tree.

Using induction, we can prove:
\begin{lemma}
There is a one-to-one correspondence between parse trees $P$ of $F$ and consistent
cycle covers $C$ of $G_1$. The monomial of $P$ equals the weight of $F$.
Furthermore, all cycles in a consistent cycle cover of $G_1$ have length at most two.
\end{lemma}

\begin{proof}
For a subformula $H$ of $F$, $G_{H}$ denotes the graph defined
at the beginning of Section~\ref{sec:G1:putting}
We prove the following more general statement:
\begin{itemize}
\item There is a one-to-one correspondence between parse trees $P$ of $H$ and consistent
cycle covers $C$ of $G_H$ not covering the top node (in the case of addition and input gates)
or not containing the self-loop at the top of the first $2$-cycle (in the case
of multiplication gates, respectively).
\item The monomial of $P$ equals the weight of $C$.
\item If $H$ is an input gate, then here is exactly one cycle cover of $G_H$ where the top node is not covered externally. This cycle cover has weight $1$.
\item If the top gate of $H$ is an addition gate,
then there is exactly one cycle cover of $G_H$ where the top node is not covered externally.
This cycle cover has weight $1$.
\item If the top gate of $H$ is a multiplication gate,
then there is exactly one cycle cover of $G_H$ where the top node of the first 2-cycle
is covered by the self-loop. This cycle cover has weight $1$.
\item All cycles in a consistent cycle cover of $G_H$ have length at most two.
\end{itemize}

The proof is by structural induction. If $H$ consists of one node, then
it is an input gate. Let $L$ be its label. 
$H$ has one parse tree with label $L$. On the other hand, there is exactly
one consistent cycle cover not covering the node at the top.
The weight of this cover is $L$. If the top node is covered,
then $C$ consists of the $2$-cycle like in
Figure~\ref{fig:input} on the right-hand side. Its weight is $1$.

If the top gate of $H$ is an addition gate, then let $H_1,\dots,F_t$ 
be its children, which have a multiplication gate at the top.
If the top node of the top addition gadget of $G_H$ is not covered, then there are $\ell$ ways
to cover the addition gate, as depicted on the right-hand side of Figure~\ref{fig:add:cover}.
Each parse tree of $H$ is a parse tree of some $H_\tau$, $1 \le \tau \le t$,
plus one additional edge. By the induction hypothesis, there is a one-to-one correspondence 
between parse trees of $H_\tau$ and consistent cycle covers of $G_{H_\tau}$ 
not containing the self-loop at the top of the first $2$-cycle. (Since the
formula is layered, the top gates of $H_1,\dots,H_t$ are multiplication gates.
Hence, there is also a one-to-one correspondence between cycle covers of $G_H$ and parse
trees of $H$, since by the induction hypothesis, there is only one cover for the subgraphs
corresponding to $H_{\tau'}$, $\tau' \not= \tau$, and they all have weight $1$.
Thus, the weight of the cover of $H$ equals the weight of the cover of $H_\tau$.
If the top node of the addition gadget of $G_H$ is covered, then we 
are in the situation of the left-hand side of 
Figure~\ref{fig:add:cover}. By the induction hypothesis, there is only one
way to cover each of the subgraphs $G_{H_\tau}$, too. The total weight of this cover is $1$.

Finally, if the top gate of $H$ is a multiplication gate with subformalas
$H_!,\dots,H_t$, then every parse
tree of $H$ consists of parse trees of $H_1,\dots,H_\ell$. 
If the first $2$-cycle is covered, then all $2$-cycles are covered.
Therefore, the top nodes of $G_{H_\tau}$, $1 \le \tau \le t$, are all covered
externally, and since the top gates of the $H_\tau$ are either addition or input
gates, there is one cover of each parse tree of $H_\tau$ and this cover has weight
equal to the corresponding monomial. The weight of the corresponding cover
of $G_H$ is the product of these weights/monomials, and therefore, the weight
equals the monomial of the parse tree. If the first $2$-cycle is not covered,
then none of the $2$-cycle is covered. Therefore, the subgraphs $G_{H_\tau}$
have only one cover and this cover has weight $1$. 

The fact that no cover has cycle of length $> 2$ follows from the fact that 
no gadget has cycles of length $> 2$.
\end{proof}

\subsection{The enumeration gadget}

We are given a formula $F(X_1,\dots,X_m,Y_1,\dots,Y_n)$
and we want to sum over the $Y$-variables. We will represent
each $Y_i$ by a directed edge $y_i =(s_i,t_i)$. These edges
will be called \emph{$y$-edges}.
There will be
directed edges from each $t_j$ to each $s_\ell$ except for $j = \ell$
connecting the $y$-edges. These edges will be called \emph{connecting edges}.
Each $s_\ell$ and $t_j$ gets a self-loop.
Call the resulting 
graph $R_n$. The graph $R_n$ has the following properties:
\begin{itemize}
    \item Each directed cycle that is not a self-loop 
    has even length, every second edge if a $y$-edge
    and every other edge is a connecting edge.
\end{itemize}

\begin{lemma}
\begin{enumerate}
    \item For every set of $k$ $y$-edges, there are $k!$ many $(2k,1)$-restricted cycle 
    covers containing these $y$-edges and no other $y$-edges.
    \item Every cycle cover that is $(2k,c)$-restricted and contains more than $k$
    $y$-edges fulfills $c \ge 4$.
\end{enumerate}    
\end{lemma}

\begin{proof}
    The $y$-edges can be visited in any order. Any two $y$-edges can be connected by a unique
    connecting edge. Thus there are $k!$ cycles of length $2k$ covering a given set of $y$-edges
    of size $k$. The remaining nodes can be covered by self-loops.

    A cycle of length $2k$ has exactly $k$ $y$-edges. Thus to cover more than $k$ $y$-edges,
    we need a second cycle. Except for the self-loops, the shortest cycles in $R_n$ have length four.
\end{proof}

\subsection{The graph $G_2$}

The graph $G_2$ is built as follows. 
\begin{itemize}
\item We take the graph $G_1$.
\item We add an enumeration gadget $R_n$.
\item Let $\ell_1,\dots,\ell_s$ be the loops of the input gadgets that are labeled
with $Y_i$. We iff-couple the $y_i$-edge of $R_n$ with $\ell_1$, $\ell_1$ with $\ell_2$, and so on.
We do so for every $1 \le i \le n$. 
\item We replace all the weights $Y_i$ by $1$.
\item Furthermore, we add a self-loop to the top-node of every input gate that was labeled with 
$Y_i$. This gives two ways to cover such a gadget when it is not in a parse-tree. One as before with a 2-cycle and the other one with two self-loops. This will be important, since selecting the loop
that corresponds to $Y_i$ means setting it to $1$, independent of whether it is in a parse-tree
or not. However, only
one of the two local covers can be chosen, depending on whether $Y_i$ is set to $1$ or not.
\item If $Y_i$ is set to $0$ and the corresponding input gate is in the parse tree,
then there is no consistent cycle cover anymore. This is all right, since the corresponding
monomial contains $Y_i$, which is set to $0$. If the input gate is not in the parse tree,
then is can locally be covered by the $2$-cycle.
\end{itemize}

\begin{lemma}
The cycle  of length $2k$ in every consistent $(2k,2)$-cycle cover of $G_2$ 
is contained in $R_n$.
\end{lemma}

\begin{proof}
The longest cycle in $G_1$ has length two. Thus, the cycle of length $2k$ can only be in $R_n$
\end{proof}

A consistent $(2k,2)$-restricted cycle cover cannot have any other cycles with $y$-edges in $R_n$.
We call two consistent $(2k,2)$-restricted cycle covers $y$-equivalent if they contain the 
same $y$-edges.

Let $F(X,Y)$ be our given formula. For a $\{0,1\}$-assignment $\eta$ to the $Y$-variables,
let $F_\eta$ denote the resulting formula. 

\begin{lemma}
    There is a bijection of parse trees of $F_\eta$ with nonzero monomial $M$ and the equivalence
    classes of the $(2k,2)$-restricted cycle covers with nonzero weight.
\end{lemma}

\begin{proof}
    There is a one-to-one correspondence between the parse trees of $F$ and the consistent cycle
    covers of $G_1$. If a parse tree $P$ has a nonzero monomial in $F_\eta$, then in $F$, the monomial can only contain $Y_i$-variables, that are set to $1$ under $\eta$.
\end{proof}

\subsection{The graph $G_3$}
\label{sec:G3}

Now the graph $G_3$ is obtained by replacing the iff-couplings with
the gadget in Figure~\ref{fig:iff}.
If the two edge $(x,y)$ and $(u,v)$ are iff-coupled, then we subdivide the
edges with the nodes $a$ and $b$ and connect them with the subgraph as depicted.
For each iff-coupling, we insert a new subgraph. % The edges now have weights.
If we do not write a weight explicitly, then the weight of the edge is $1$.
Similar gadgets were developed in the past, see e.g.\ \cite{DBLP:journals/talg/DellHMTW14}.
The difference in our gadget is that we have a $4$-cycle between $a$ and $b$ instead of 
a $2$-cycle. This will be crucial, since $(k,c)$-restricted cycle covers are sensitive
to changes of cycle lengths.

\begin{figure}
\begin{minipage}{0.45\textwidth}
\centering
\begin{tikzpicture}[scale=0.5,->]
    \node[vertex] (u) at (5,6) {$u$};
    \node[vertex] (v) at (5,0) {$v$};
    \node[vertex] (x) at (0,6) {$x$};
    \node[vertex] (y) at (0,0) {$y$};
    \node[vertex] (a) at (0,3) {$a$};
    \node[vertex] (b) at (5,3) {$b$};
    \node[vertex] (c) at (2.5,5.5) {$c$};
    \node[vertex] (d) at (2.5,3.5) {$d$};
    \node[vertex] (e) at (2.5,2.5) {$e$};
    \path (x) edge (a);
    \path (a) edge[loop left] (a); 
    \path (a) edge node[left] {$-2$}  (y);
    \path (u) edge (b);
    \path (b) edge[loop right] (b); 
    \path (b) edge (v);
    \path (a) edge (d);
    \path (d) edge (b);
    \path (b) edge (e);
    \path (e) edge (a);
    \path (a) edge[bend left = 20] (c);
    \path (c) edge[bend left = 20] (a); 
    \path (b) edge[bend left = 20] (c);
    \path (c) edge[bend left = 20] (b);
    \path (c) edge[loop above] node {$-1$} (c);
    \path (d) edge[loop above] (d);
    \path (e) edge[loop below] (e);
\end{tikzpicture}
\caption{The iff-gadget. The edges $(x,y)$ and $(u,v)$ are the iff-coupled edges in
the original graph. \label{fig:iff}}
\end{minipage}\hspace*{0.1\textwidth}\begin{minipage}{0.45\textwidth}
\centering
\begin{tikzpicture}[scale=0.5,->]
    \node[vertex] (u) at (5,6) {$u$};
    \node[vertex] (v) at (5,0) {$v$};
    \node[vertex] (x) at (0,6) {$x$};
    \node[vertex] (y) at (0,0) {$y$};
    \node[vertex] (a) at (0,3) {$a$};
    \node[vertex] (b) at (5,3) {$b$};
    \node[vertex] (c) at (2.5,5.5) {$c$};
    \node[vertex] (d) at (2.5,3.5) {$d$};
    \node[vertex] (e) at (2.5,2.5) {$e$};
    \path (x) edge[color = blue] (a);
    \path (a) edge[loop left] (a); 
    \path (a) edge[color = blue] node[left] {$-2$}  (y);
    \path (u) edge[color = blue] (b);
    \path (b) edge[loop right] (b); 
    \path (b) edge[color = blue] (v);
    \path (a) edge (d);
    \path (d) edge (b);
    \path (b) edge (e);
    \path (e) edge (a);
    \path (a) edge[bend left = 20] (c);
    \path (c) edge[bend left = 20] (a); 
    \path (b) edge[bend left = 20] (c);
    \path (c) edge[bend left = 20] (b);
    \path (c) edge[loop above, color = blue] node {$-1$} (c);
    \path (d) edge[loop above, color = blue] (d);
    \path (e) edge[loop below, color = blue] (e);
\end{tikzpicture}
\caption{The covering of the iff-gadget if both edges $(x,y)$ and $(u,v)$ 
appear in the original cycle cover. \label{fig:iff:consistent:1}}
\end{minipage}
\end{figure}

\subsubsection{Local coverings of the iff-gadget}

There are essentially four different cases how an iff-gadget can be covered:

\begin{itemize}
\item If both edges are taken in $G_2$, then there is one way to cover the iff-coupling 
internally, drawn in blue in Figure~\ref{fig:iff:consistent:1}.
The contribution to the overall weight of a cover is $(-2) \cdot (-1) = 2$.
\item If both edges are not taken, then there are six ways how to cover the gadget locally,
shown in Figure~\ref{fig:iff:consistent:2}.
Two of them have weight $-1$, four have weight $1$. The overall contribution
to the weight is $2$.
\item If one edge is taken but the other one is not, then there are two ways to cover the gadget.
These covers have opposite sign. See Figure~\ref{fig:iff:inconsistent:1}.
The situation when the other edge is taken is symmetric.
\item Then there is finally the case when the gadget is covered inconsistently, that is, it is entered via $x$ and left via $v$. Again there are two covers with opposite signs, see
Figure~\ref{fig:iff:inconsistent:2}. Again, there is a symmetric case.
\end{itemize}

\begin{figure}
\centering
\begin{tikzpicture}[scale=0.5,->]
    \node[vertex] (u) at (5,6) {$u$};
    \node[vertex] (v) at (5,0) {$v$};
    \node[vertex] (x) at (0,6) {$x$};
    \node[vertex] (y) at (0,0) {$y$};
    \node[vertex] (a) at (0,3) {$a$};
    \node[vertex] (b) at (5,3) {$b$};
    \node[vertex] (c) at (2.5,5.5) {$c$};
    \node[vertex] (d) at (2.5,3.5) {$d$};
    \node[vertex] (e) at (2.5,2.5) {$e$};
    \path (x) edge (a);
    \path (a) edge[loop left, color = blue] (a); 
    \path (a) edge node[left] {$-2$}  (y);
    \path (u) edge (b);
    \path (b) edge[loop right, color = blue] (b); 
    \path (b) edge (v);
    \path (a) edge (d);
    \path (d) edge (b);
    \path (b) edge (e);
    \path (e) edge (a);
    \path (a) edge[bend left = 20] (c);
    \path (c) edge[bend left = 20] (a); 
    \path (b) edge[bend left = 20] (c);
    \path (c) edge[bend left = 20] (b);
    \path (c) edge[loop above, color = blue] node {$-1$} (c);
    \path (d) edge[loop above, color = blue] (d);
    \path (e) edge[loop below, color = blue] (e);
\end{tikzpicture}~~~~
\begin{tikzpicture}[scale=0.5,->]
    \node[vertex] (u) at (5,6) {$u$};
    \node[vertex] (v) at (5,0) {$v$};
    \node[vertex] (x) at (0,6) {$x$};
    \node[vertex] (y) at (0,0) {$y$};
    \node[vertex] (a) at (0,3) {$a$};
    \node[vertex] (b) at (5,3) {$b$};
    \node[vertex] (c) at (2.5,5.5) {$c$};
    \node[vertex] (d) at (2.5,3.5) {$d$};
    \node[vertex] (e) at (2.5,2.5) {$e$};
    \path (x) edge (a);
    \path (a) edge[loop left] (a); 
    \path (a) edge node[left] {$-2$}  (y);
    \path (u) edge (b);
    \path (b) edge[loop right] (b); 
    \path (b) edge (v);
    \path (a) edge[color = blue] (d);
    \path (d) edge[color = blue] (b);
    \path (b) edge[color = blue] (e);
    \path (e) edge[color = blue] (a);
    \path (a) edge[bend left = 20] (c);
    \path (c) edge[bend left = 20] (a); 
    \path (b) edge[bend left = 20] (c);
    \path (c) edge[bend left = 20] (b);
    \path (c) edge[loop above, color = blue] node {$-1$} (c);
    \path (d) edge[loop above] (d);
    \path (e) edge[loop below] (e);
\end{tikzpicture}~~~~
\begin{tikzpicture}[scale=0.5,->]
    \node[vertex] (u) at (5,6) {$u$};
    \node[vertex] (v) at (5,0) {$v$};
    \node[vertex] (x) at (0,6) {$x$};
    \node[vertex] (y) at (0,0) {$y$};
    \node[vertex] (a) at (0,3) {$a$};
    \node[vertex] (b) at (5,3) {$b$};
    \node[vertex] (c) at (2.5,5.5) {$c$};
    \node[vertex] (d) at (2.5,3.5) {$d$};
    \node[vertex] (e) at (2.5,2.5) {$e$};
    \path (x) edge (a);
    \path (a) edge[loop left, color = blue] (a); 
    \path (a) edge node[left] {$-2$}  (y);
    \path (u) edge (b);
    \path (b) edge[loop right] (b); 
    \path (b) edge (v);
    \path (a) edge (d);
    \path (d) edge (b);
    \path (b) edge (e);
    \path (e) edge (a);
    \path (a) edge[bend left = 20] (c);
    \path (c) edge[bend left = 20] (a); 
    \path (b) edge[bend left = 20, color = blue] (c);
    \path (c) edge[bend left = 20, color = blue] (b);
    \path (c) edge[loop above] node {$-1$} (c);
    \path (d) edge[loop above, color = blue] (d);
    \path (e) edge[loop below, color = blue] (e);
\end{tikzpicture}

\begin{tikzpicture}[scale=0.5,->]
    \node[vertex] (u) at (5,6) {$u$};
    \node[vertex] (v) at (5,0) {$v$};
    \node[vertex] (x) at (0,6) {$x$};
    \node[vertex] (y) at (0,0) {$y$};
    \node[vertex] (a) at (0,3) {$a$};
    \node[vertex] (b) at (5,3) {$b$};
    \node[vertex] (c) at (2.5,5.5) {$c$};
    \node[vertex] (d) at (2.5,3.5) {$d$};
    \node[vertex] (e) at (2.5,2.5) {$e$};
    \path (x) edge (a);
    \path (a) edge[loop left] (a); 
    \path (a) edge node[left] {$-2$}  (y);
    \path (u) edge (b);
    \path (b) edge[loop right, color = blue] (b); 
    \path (b) edge (v);
    \path (a) edge (d);
    \path (d) edge (b);
    \path (b) edge (e);
    \path (e) edge (a);
    \path (a) edge[bend left = 20, color = blue] (c);
    \path (c) edge[bend left = 20, color = blue] (a); 
    \path (b) edge[bend left = 20] (c);
    \path (c) edge[bend left = 20] (b);
    \path (c) edge[loop above] node {$-1$} (c);
    \path (d) edge[loop above, color = blue] (d);
    \path (e) edge[loop below, color = blue] (e);
\end{tikzpicture}~~~~
\begin{tikzpicture}[scale=0.5,->]
    \node[vertex] (u) at (5,6) {$u$};
    \node[vertex] (v) at (5,0) {$v$};
    \node[vertex] (x) at (0,6) {$x$};
    \node[vertex] (y) at (0,0) {$y$};
    \node[vertex] (a) at (0,3) {$a$};
    \node[vertex] (b) at (5,3) {$b$};
    \node[vertex] (c) at (2.5,5.5) {$c$};
    \node[vertex] (d) at (2.5,3.5) {$d$};
    \node[vertex] (e) at (2.5,2.5) {$e$};
    \path (x) edge (a);
    \path (a) edge[loop left] (a); 
    \path (a) edge node[left] {$-2$}  (y);
    \path (u) edge (b);
    \path (b) edge[loop right] (b); 
    \path (b) edge (v);
    \path (a) edge (d);
    \path (d) edge (b);
    \path (b) edge[color = blue] (e);
    \path (e) edge[color = blue] (a);
    \path (a) edge[bend left = 20, color = blue] (c);
    \path (c) edge[bend left = 20] (a); 
    \path (b) edge[bend left = 20] (c);
    \path (c) edge[bend left = 20, color = blue] (b);
    \path (c) edge[loop above] node {$-1$} (c);
    \path (d) edge[loop above, color = blue] (d);
    \path (e) edge[loop below] (e);
\end{tikzpicture}~~~~
\begin{tikzpicture}[scale=0.5,->]
    \node[vertex] (u) at (5,6) {$u$};
    \node[vertex] (v) at (5,0) {$v$};
    \node[vertex] (x) at (0,6) {$x$};
    \node[vertex] (y) at (0,0) {$y$};
    \node[vertex] (a) at (0,3) {$a$};
    \node[vertex] (b) at (5,3) {$b$};
    \node[vertex] (c) at (2.5,5.5) {$c$};
    \node[vertex] (d) at (2.5,3.5) {$d$};
    \node[vertex] (e) at (2.5,2.5) {$e$};
    \path (x) edge (a);
    \path (a) edge[loop left] (a); 
    \path (a) edge node[left] {$-2$}  (y);
    \path (u) edge (b);
    \path (b) edge[loop right] (b); 
    \path (b) edge (v);
    \path (a) edge[color = blue] (d);
    \path (d) edge[color = blue] (b);
    \path (b) edge (e);
    \path (e) edge (a);
    \path (a) edge[bend left = 20] (c);
    \path (c) edge[bend left = 20, color = blue] (a); 
    \path (b) edge[bend left = 20, color = blue] (c);
    \path (c) edge[bend left = 20] (b);
    \path (c) edge[loop above] node {$-1$} (c);
    \path (d) edge[loop above] (d);
    \path (e) edge[loop below, color = blue] (e);
\end{tikzpicture}
\caption{The six ways to cover an iff-gadget consistently, when both edges $(x,y)$ and $(u,v)$
are not in the original cycle cover. \label{fig:iff:consistent:2}}
\end{figure}
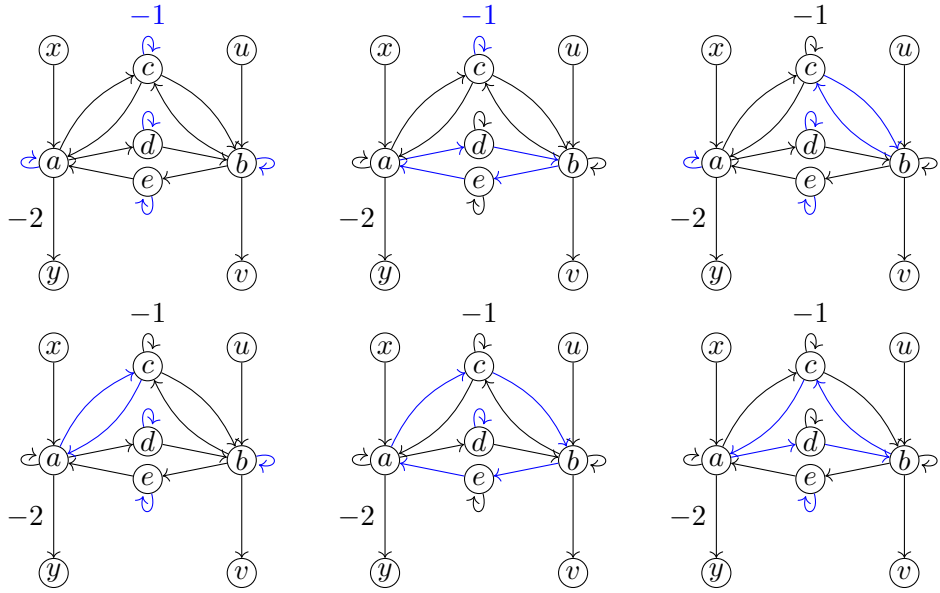

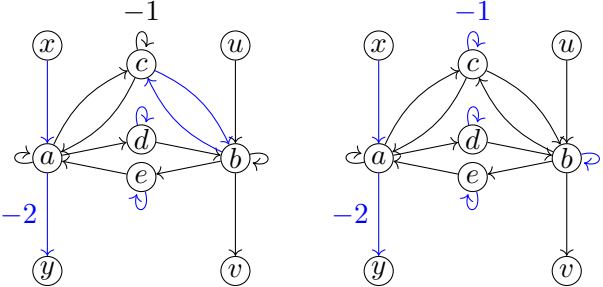
\begin{figure}
\centering
\begin{tikzpicture}[scale=0.5,->]
    \node[vertex] (u) at (5,6) {$u$};
    \node[vertex] (v) at (5,0) {$v$};
    \node[vertex] (x) at (0,6) {$x$};
    \node[vertex] (y) at (0,0) {$y$};
    \node[vertex] (a) at (0,3) {$a$};
    \node[vertex] (b) at (5,3) {$b$};
    \node[vertex] (c) at (2.5,5.5) {$c$};
    \node[vertex] (d) at (2.5,3.5) {$d$};
    \node[vertex] (e) at (2.5,2.5) {$e$};
    \path (x) edge[color = blue] (a);
    \path (a) edge[loop left] (a); 
    \path (a) edge[color = blue] node[left] {$-2$}  (y);
    \path (u) edge (b);
    \path (b) edge[loop right] (b); 
    \path (b) edge (v);
    \path (a) edge (d);
    \path (d) edge (b);
    \path (b) edge (e);
    \path (e) edge (a);
    \path (a) edge[bend left = 20] (c);
    \path (c) edge[bend left = 20] (a); 
    \path (b) edge[bend left = 20, color = blue] (c);
    \path (c) edge[bend left = 20, color = blue] (b);
    \path (c) edge[loop above] node {$-1$} (c);
    \path (d) edge[loop above, color = blue] (d);
    \path (e) edge[loop below, color = blue] (e);
\end{tikzpicture}~~~~
\begin{tikzpicture}[scale=0.5,->]
    \node[vertex] (u) at (5,6) {$u$};
    \node[vertex] (v) at (5,0) {$v$};
    \node[vertex] (x) at (0,6) {$x$};
    \node[vertex] (y) at (0,0) {$y$};
    \node[vertex] (a) at (0,3) {$a$};
    \node[vertex] (b) at (5,3) {$b$};
    \node[vertex] (c) at (2.5,5.5) {$c$};
    \node[vertex] (d) at (2.5,3.5) {$d$};
    \node[vertex] (e) at (2.5,2.5) {$e$};
    \path (x) edge[color = blue] (a);
    \path (a) edge[loop left] (a); 
    \path (a) edge[color = blue] node[left] {$-2$}  (y);
    \path (u) edge (b);
    \path (b) edge[loop right, color = blue] (b); 
    \path (b) edge (v);
    \path (a) edge (d);
    \path (d) edge (b);
    \path (b) edge (e);
    \path (e) edge (a);
    \path (a) edge[bend left = 20] (c);
    \path (c) edge[bend left = 20] (a); 
    \path (b) edge[bend left = 20] (c);
    \path (c) edge[bend left = 20] (b);
    \path (c) edge[loop above, color = blue] node {$-1$} (c);
    \path (d) edge[loop above, color = blue] (d);
    \path (e) edge[loop below, color = blue] (e);
\end{tikzpicture}
\caption{The two ways to cover an iff-gadget if one edge $(x,y)$ is in the original cover
and the other one $(u,v)$ is not. Both covers have opposite signs.
\label{fig:iff:inconsistent:1}}
\end{figure}

\begin{figure}
\centering
\begin{tikzpicture}[scale=0.5,->]
    \node[vertex] (u) at (5,6) {$u$};
    \node[vertex] (v) at (5,0) {$v$};
    \node[vertex] (x) at (0,6) {$x$};
    \node[vertex] (y) at (0,0) {$y$};
    \node[vertex] (a) at (0,3) {$a$};
    \node[vertex] (b) at (5,3) {$b$};
    \node[vertex] (c) at (2.5,5.5) {$c$};
    \node[vertex] (d) at (2.5,3.5) {$d$};
    \node[vertex] (e) at (2.5,2.5) {$e$};
    \path (x) edge[color = blue] (a);
    \path (a) edge[loop left] (a); 
    \path (a) edge node[left] {$-2$}  (y);
    \path (u) edge (b);
    \path (b) edge[loop right] (b); 
    \path (b) edge[color = blue] (v);
    \path (a) edge (d);
    \path (d) edge (b);
    \path (b) edge (e);
    \path (e) edge (a);
    \path (a) edge[bend left = 20, color = blue] (c);
    \path (c) edge[bend left = 20] (a); 
    \path (b) edge[bend left = 20] (c);
    \path (c) edge[bend left = 20, color = blue] (b);
    \path (c) edge[loop above] node {$-1$} (c);
    \path (d) edge[loop above, color = blue] (d);
    \path (e) edge[loop below, color = blue] (e);
\end{tikzpicture}~~~~
\begin{tikzpicture}[scale=0.5,->]
    \node[vertex] (u) at (5,6) {$u$};
    \node[vertex] (v) at (5,0) {$v$};
    \node[vertex] (x) at (0,6) {$x$};
    \node[vertex] (y) at (0,0) {$y$};
    \node[vertex] (a) at (0,3) {$a$};
    \node[vertex] (b) at (5,3) {$b$};
    \node[vertex] (c) at (2.5,5.5) {$c$};
    \node[vertex] (d) at (2.5,3.5) {$d$};
    \node[vertex] (e) at (2.5,2.5) {$e$};
    \path (x) edge[color = blue] (a);
    \path (a) edge[loop left] (a); 
    \path (a) edge node[left] {$-2$}  (y);
    \path (u) edge (b);
    \path (b) edge[loop right] (b); 
    \path (b) edge[color = blue] (v);
    \path (a) edge[color = blue] (d);
    \path (d) edge[color = blue] (b);
    \path (b) edge (e);
    \path (e) edge (a);
    \path (a) edge[bend left = 20] (c);
    \path (c) edge[bend left = 20] (a); 
    \path (b) edge[bend left = 20] (c);
    \path (c) edge[bend left = 20] (b);
    \path (c) edge[loop above, color = blue] node {$-1$} (c);
    \path (d) edge[loop above] (d);
    \path (e) edge[loop below, color = blue] (e);
\end{tikzpicture}
\caption{The two ways to cover an iff-gadget if 
the gadget is entered on the one side and left on the other. 
\label{fig:iff:inconsistent:2}}
\end{figure}
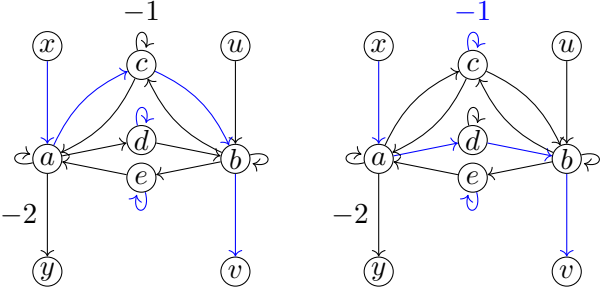

\subsubsection{Consistent cycle covers}

A consistent cycle cover $C$ of $G_2$ are mapped to cycle covers of $G_3$ where each
iff-gadget is covered consistently. If both edges of an iff-gadget 
are taken, then there is one way to cover the gadget internally. 
This gives a multiplicative factor of $2$.
If both edges are not taken, then there are six ways to cover the gadget
internally. Again, the overall contribution is $2$.
If there are $M$ gadgets in total, then $C$ will get mapped to a bunch
of cycle cover in this way with total weight $2^M w(C)$.

If the cycle cover $C$ is $(2k,2)$-restricted, the resulting cycle covers
will be $(3k,4)$-restricted, since each iff-coupled edge is subdivided.
Each edge is only subdivided once except for the loops at the input
gates that were labeled with a $Y$-variable. These are subdivided twice,
yielding a cycle of length $3$.
Furthermore, the internal cycles of the iff-gadgets have length at most $4$.

On the other hand, if there is a $(3k,4)$-restricted cycle cover such that all iff-gadget
are covered consistently, then this corresponds to exactly one $(2k,2)$-restricted
consistent cycle cover of $G_2$. The cycle of length $3k$ will be contained
in the $R_n$-part of $G_3$.

\subsubsection{Inconsistent cycle covers}

To get rid of the inconsistent cycle covers, we define an involution on the set of inconsistent cycle covers. A cycle cover is inconsistent if at least one iff-gadget is not covered consistently. We define an involution on the set of all inconsistent cycle covers as follows: We number the iff-gadgets arbitrarily. Let $C$ be an inconsistent cycle cover and let $I$ be the first iff gadget that is not covered consistently. We map $C$ to the cycle cover $C'$ where $I$ is covered in the other way as depicted in the Figures~\ref{fig:iff:inconsistent:1} and~\ref{fig:iff:inconsistent:2}.
This new cycle cover has weight $w(C') = - w(C)$. The mapping $C \mapsto C'$ is an involution by construction. Finally, if $C$ is $(k,c)$-restricted for some $c \ge 2$,
then $C'$ is also $(k,c)$-restricted. This is obvious in the first case, since here only 
the local covering is changed. In the second case, the length of the path that crosses $I$ is not changed (this is the important change to the gadget!), and therefore all cycle-length stay the same. Thus, the overall contribution of the inconsistent cycle covers sums up to $0$.

Altogether, this proves the main result of the present section.

\begin{theorem}
Let $F$ be a layered formula in variables $X_1,\dots,X_m$ and $Y_1,\dots,Y_n$.
Let $G_3$ be the graph defined as above.
%Let $P$ be the $(3k,4)$-restricted permanent of $G_3$. 
Then 
\[
  \per^{(3k,\le 4)}(G_3) = k! \cdot 2^M \cdot \sum_{e \in \ones{n}{k}} F(X,e)
\]
where $M$ is the number of iff-couplings 
\end{theorem}

\begin{remark}
The factor $k!$ comes from the number of $3k$-cycles in the $R_n$-part.
This can be avoided by letting `all connecting edges $(t_i,s_j)$ with $i > j$
go through the same new node $b$. In this way, the $y$-edges have to be visited
in ascending order. The factor $2^M$ seems to be unavoidable though.
\end{remark}

\begin{proof}[Proof of Theorem~\ref{cor:kcper:hard}] 
Let $F_n$ be a universal formula that can simulate any formula of size $\le n$.
%\mdc{reference needed}. 
Consider the graph $G_1$ (see Section~\ref{sec:G1})
and replace the iff-couplings of the multiplication
gadgets by an edge, that subdivides and connects the iff-coupled edges of $G_1$.
The resulting graph is essentially a tree that has some $2$-cycles and $4$-cycles.
It is easy to see that the treewidth of this graph is $2$. Now in $G_3$,
instead of the edges between iff-coupled edges, we have the iff-gadget. 
They introduce $3$ more nodes, therefore, the treewidth of this part of $G_3$ 
is bounded by $5$. The other part of $G_3$ has girth $> 4$.

Thus, the family $H_n$ will be the graphs $G_3$ corresponding to $F_n$.
By our construction, every family in  $\VW[\sfF]$ is reducible to this family.
\end{proof}

\section{Upper bound}
\label{sec:upper}

\begin{lemma} \label{lem:kcrestricted:1}
Let $G = (V_1 \cup V_2, E)$ be a $(c,b)$-nice graph and $C$ be a $(k,c)$-restricted cycle cover.
Let $c$ be the cycle of length $k$ in $C$. Then all nodes of $V_1$ that 
are not in $c$ are covered by self-loops in $C$.
\end{lemma}

\begin{proof}
Since $G[V_1]$ has girth $> c$ (except for self-loops), the only cycles of length
$\le c$ in $G[V_1]$ are self-loops.
\end{proof}

Let $G$ be an arbitrary edge-weighted graph. We define $\per^{(\le c)}(G) = \sum_C w(C)$ where
the sum is taken over all cycle covers with all cycles having length $\le c$.

\begin{theorem} \label{thm:boundedtw:per}
   Let $G$ be a graph of bounded tree-width. Then there is an algebraic
   circuit of fpt size that computes $\per^{(\le c)}(G)$.
\end{theorem}

\begin{proof}
Bodlaender and Hagerup \cite{DBLP:conf/icalp/BodlaenderH95} show that whenever
a graph has bounded tree-width, then there is a binary tree-decomposition
of logarithmic height. Moreover, we can assume that the tree decomposition is nice,
see e.g. \cite{cygan2015parameterized} for a definition, and the height is still logarithmic.

Let $T = (I,F)$ be a nice tree decomposition of $G$ of logarithmic height.
For each node $i \in I$, let $V_i$ be the set of all nodes that appear in the 
subtree below $i$ but not in $X_i$.  

A \emph{path-cycle cover} of a graph is a collection of node disjoint path and cycles.
Following the tree-decomposition, we will construct inductively path-cycle covers.
Eventually, all paths need to be closed to a cycle in the computation
of $\per^{(\le c)}(G)$.

For each node $i$, we construct circuits computing certain polynomials $P_{i,C}$
with $C$ being a path-cycle cover containing all nodes of $X_i$ and potentially
some nodes from $V_i$, however each path or cycle has to contain
at least one node of $X_i$. 
Each path has length $\le c-2$ and each cycle has length $\le c$.
The cover $C$ contains a constant number of nodes, since the path and cycles
have length bounded by a constant.\footnote{Note, however, that the bound
on the treewidth is what matters.
If the bound $c$ was not constant,
then the problem would still be fpt. However, the description
of the construction would be more complicated.}
In the cover, we treat uncovered nodes as path of length $0$.
We construct the polynomials inductively: 
\begin{itemize}
\item If the node $i$ is a leaf, then $X_i$ is empty and there is only one polynomial
$P_{i,\emptyset} = 1$.
\item If $i$ is an introduce node, let $x$ the introduced node. Let $P_{j,D}$ 
be a polynomial computed at the (unique) child $j$ of $i$. 
For each such polynomial, there might be several ways how $x$
can be added to the cover $D$ yielding a new cover $C$.
Each such new cover $C$ gives a polynomial $P_{i,C}$:
\begin{itemize}
    \item $P_{i,C} = P_{j,D}$, 
    where $C$ is obtained from $D$ by adding the path $x$ of length $0$.
    \item $P_{i,C} = w(x,u) \cdot P_{j,D}$ 
    for each $u$ such that there is a path $p$ starting in 
    $u$ of length $\le c - 3$ and there is an edge $(x,u)$.
    The cover $C$ is obtained from $D$ by prepending $x$ to $p$. 
    \item $P_{i,C} = w(v,x) \cdot P_{j,D}$ 
    for each $v$ such that there is a path $p$ ending in 
    $v$ of length $\le c - 3$ and there is an edge $(x,u)$.
    The cover $C$ is obtained from $D$ by appending $x$ to $p$. 
    \item $P_{i,C} = w(v,x) w(x,u')  P_{j,D}$
    for all paths $p$ ending in $v$ and paths $p'$ starting in $u'$
    such that there are edges $(v,x)$ and $(x,u')$ and the total length
    of the resulting path is $\le c-2$. $C$ is obtained from $D$ by connecting
    $p$ and $q$ using $x$.
    \item $P_{i,C} = w(v,x) w(x,u)  P_{j,D}$
    for each path $p$ from $u$ to $v$ of lenght $\le c-2$
    such that there are edges $(v,x)$ and $(x,u)$.
    $C$ is obtained from $D$ by closing the path $p$ using $x$.
    \item $P_{i,C} = w(x,x) \cdot P_{j,D}$ if $(x,x)$ is an edge, 
    $C$ is obtained from $D$ by adding a self-loop. 
\end{itemize}
\item If $i$ is a forget node, then $P_{i,C} = P_{j,D}$
if $x$ is covered by a cycle $c$ in $D$ or $x$ is not the start or end node
of a path $p$.
If there are not any nodes of $c$ 
still in $X_i$, then we remove $c$ from $D$ to obtain $C$. Otherwise, $C = D$.
If $x$ is the start or end node of a path, then we simply
drop $P_{j,D}$, since we cannot cover $x$ by a cycle after it is forgotten.
\item If $i$ is a join node, let $P_{j,D}$ and $P_{j',D'}$ 
denote the polynomials computed
at the two children $j$ and $j'$ of $i$. For a given cover $C$ at $i$,
we have
\[
   P_{i,C} = \sum_{D,D'} P_{j,D} \cdot P_{j',D'}
\]
where $D$ and $D'$ run over all covers such that all cycles and all path
of length $> 0$ that only contain nodes of $X_i$ appear in $D$,
all cycles that contain nodes of $V_j$ appear in $D$ and all cycles that contain
nodes of $V_{j'}$ appear in $D'$. Note that path and cycles that only contain
nodes of $X_i$ could also appear in $D'$; by forcing them to appear in $D$, we
make the decomposition of $C$ into $D$ and $D'$ unique.
\item Finally, if $i$ is the root, then $X_i = \emptyset$. The child $j$ of $i$ is a forget 
node. We set $P_{i,\emptyset} = P_{j,\emptyset}$.
\end{itemize}
From the construction it is clear that the polynomial computed at the root is
the restricted permanent $\per^{(\le c)}(G)$. By following the tree decomposition
from the leaves to the root, we get an algebraic circuit of fpt size. 
\end{proof}

\begin{remark} \label{rem:boundedtw:per}
We can expand the algebraic circuit constructed in Theorem~\ref{thm:boundedtw:per}
into a formula. Since the tree decomposition has only logarithmic height,
the size of the formula will be $f(b)^{O(\log n)} = n^{O(\log f(b))}$ 
where $b$ is the treewidth and $f(b)$ is some function of $b$.
\end{remark}

\newcommand{\Cyc}{\operatorname{Cyc}}

\begin{proof}[Proof of Theorem~\ref{thm:kcper:easy}]
Let $G_n = (V, E)$, $|V| = n$ with a partition of $V = V_1 \cup V_2$
as in Definition~\ref{def:cbnice}, $n_i = |V_i|$. 
Consider a $(k,c)$-restricted cycle cover $C$
of $G_n$. Let $c_1$ be the cycle of length $k$ in $C$.
Then by Lemma~\ref{lem:kcrestricted:1} all nodes in $V_1$ that are not 
covered by $c_1$ are self-loops. This suggests the following
approach. We enumerate all sets of size $k$, check whether they
form a cycle. If yes, we cover the remaining nodes in $V_1$ by 
self-loops. The remaining nodes induce a graph of bounded
tree-width, and we can use Theorem~\ref{thm:boundedtw:per} 
and even Remark~\ref{rem:boundedtw:per}.

We have variables $E_{i,j}$, $1 \le i,j \le n$ representing the edges of the 
graph. We select $k$ of them using the bounded summation representing the cycle
of length $k$. We first construct a 
polynomial $\Cyc(E)$ such that $\Cyc(e) = 1$ if $e \in \{0,1\}^{n \times n}$
is the adjacency matrix of a $k$-cycle and $\Cyc(e) = 0$ otherwise.
Since there is a Boolean formula of polynomial-size which checks this,
we get an algebraic formula for $\Cyc$ of polynomial-size 
by arithmetizing the Boolean circuit.

Furthermore, we have vertex variables $Y_1,\dots,Y_n$. $Y_i$ will be set to $1$
if the corresponding node is in the $k$-cycle and to $0$ otherwise. This can
be achieved by arithmetizing $\bigvee_{j = 1}^n E_{j,i} \implies Y_i$.
We can assume that $V_1 = \{1,\dots,n_1\}$ and $V_2 = \{n_1 + 1,\dots,n\}$.
The $k$-cycle contributes weight $\prod_{i,j} E_{i,j} \cdot w(i,j)$.
The uncovered nodes in $V_1$ contribute weight 
$\prod_{i = 1}^{n_1} (1-Y_i) w(i,i)$. The weight of the uncovered nodes in $V_2$
can be in principle computed using $\per^{(\le c)}(G[V_2])$, which has
a small circuit by Theorem~\ref{thm:boundedtw:per} and
even a polynomial-size formula by Remark~\ref{rem:boundedtw:per}.
however, some nodes in $V_2$ may be covered by the $k$-cycle.
Therefore, we replace every weight $w(i,j)$ by $(1-Y_i) \cdot w(i,j)$
for $i \not= j$, turning each node $i$ off that is in the $k$-cycle.
Furthermore, we replace $w(i,i)$ by $(1-Y_i) \cdot w(i,i) + Y_i$.
This equips every node $i$ that is turned off with a self-loop with weight $1$,
ensuring that it does not contribute to $\per^{(\le c)}$. Altogether, we can
write
\[
  \sum_{e,y \in \ones{n^2}{k}}  \Cyc(e) \cdot \prod_{i,j} e_{i,j} \cdot w(i,j) \cdot 
  \left[\bigvee_{j = 1}^n e_{j,i} \implies y_i\right] 
  \cdot \prod_{i = 1}^{n_1} ( 1 - y_i) w(i,i)
  \cdot \per^{(\le c)} (G'[V_2])
\]
where $[\dots]$ denotes the arithmetization and $G'$
is the graph with the modified weight functions as described above.
\end{proof}

\bibliography{vfpt}

@inproceedings{DBLP:conf/icalp/BhattacharjeeBD24,
  author       = {Somnath Bhattacharjee and
                  Markus Bl{\"{a}}ser and
                  Pranjal Dutta and
                  Saswata Mukherjee},
  editor       = {Karl Bringmann and
                  Martin Grohe and
                  Gabriele Puppis and
                  Ola Svensson},
  title        = {Exponential Lower Bounds via Exponential Sums},
  booktitle    = {51st International Colloquium on Automata, Languages, and Programming,
                  {ICALP} 2024, Tallinn, Estonia, July 8-12, 2024},
  series       = {LIPIcs},
  volume       = {297},
  pages        = {24:1--24:20},
  publisher    = {Schloss Dagstuhl - Leibniz-Zentrum f{\"{u}}r Informatik},
  year         = {2024},
  url          = {https://doi.org/10.4230/LIPIcs.ICALP.2024.24},
  doi          = {10.4230/LIPICS.ICALP.2024.24},
  bibsource    = {dblp computer science bibliography, https://dblp.org}
}

@article{HBA,
    author = {Hesse, William and Allender, Eric and Barrington, D.A.},
    title = {Uniform constant-depth threshold circuits for division and iterated multiplication},
    journal = {Journal of Computer and System Sciences},
    year = {2002},
    volume = {65(4)},
    number={},
    pages={695-716},
    doi={10.1016/S0022-0000(02)00025-9}
}

@article{article,
author = {Bürgisser, Peter},
year = {2009},
month = {04},
pages = {81-103},
title = {On Defining Integers And Proving Arithmetic Circuit Lower Bounds},
volume = {18},
journal = {computational complexity},
doi = {10.1007/s00037-009-0260-x}
}

@InProceedings{blser_et_al:LIPIcs:2019:11464,
  author =	{Markus Bl{\"a}ser and Christian Engels},
  title =	{{Parameterized Valiant's Classes}},
  booktitle =	{14th International Symposium on Parameterized and Exact Computation (IPEC 2019)},
  pages =	{3:1--3:14},
  series =	{Leibniz International Proceedings in Informatics (LIPIcs)},
  ISBN =	{978-3-95977-129-0},
  ISSN =	{1868-8969},
  year =	{2019},
  volume =	{148},
  editor =	{Bart M. P. Jansen and Jan Arne Telle},
  publisher =	{Schloss Dagstuhl--Leibniz-Zentrum fuer Informatik},
  address =	{Dagstuhl, Germany},
  URL =		{https://drops.dagstuhl.de/opus/volltexte/2019/11464},
  URN =		{urn:nbn:de:0030-drops-114648},
  doi =		{10.4230/LIPIcs.IPEC.2019.3}
}

@inproceedings{DBLP:conf/stoc/Valiant79a,
  author       = {Leslie G. Valiant},
  editor       = {Michael J. Fischer and
                  Richard A. DeMillo and
                  Nancy A. Lynch and
                  Walter A. Burkhard and
                  Alfred V. Aho},
  title        = {Completeness Classes in Algebra},
  booktitle    = {Proceedings of the 11h Annual {ACM} Symposium on Theory of Computing,
                  April 30 - May 2, 1979, Atlanta, Georgia, {USA}},
  pages        = {249--261},
  publisher    = {{ACM}},
  year         = {1979},
  url          = {https://doi.org/10.1145/800135.804419},
  doi          = {10.1145/800135.804419},
  bibsource    = {dblp computer science bibliography, https://dblp.org}
}

@article{DBLP:journals/jc/MalodP08,
  author    = {Guillaume Malod and
               Natacha Portier},
  title     = {Characterizing {V}aliant's algebraic complexity classes},
  journal   = {J. Complexity},
  volume    = {24},
  number    = {1},
  pages     = {16--38},
  year      = {2008},
  url       = {https://doi.org/10.1016/j.jco.2006.09.006},
  doi       = {10.1016/j.jco.2006.09.006},
  bibsource = {dblp computer science bibliography, http://dblp.org}
}

@book{Buergisser:00,
  author =	 {B{\"u}rgisser, Peter},
  title =	 {Completeness and Reduction in Algebraic Complexity
                  Theory},
  publisher =	 {Springer},
  year =	 {2000}
}

@book{voll,
    author = {Vollmer, Heribert},
    title = {Introduction to Circuit Complexity},
    subtitle={A Uniform Approach},
    publisher = {Springer Berlin, Heidelberg},
    year = {1999},
    doi={10.1007/978-3-662-03927-4},
}

@book{DBLP:series/txtcs/FlumG06,
  author    = {J{\"{o}}rg Flum and
               Martin Grohe},
  title     = {Parameterized Complexity Theory},
  series    = {Texts in Theoretical Computer Science. An {EATCS} Series},
  publisher = {Springer},
  year      = {2006},
  url       = {https://doi.org/10.1007/3-540-29953-X},
  doi       = {10.1007/3-540-29953-X},
  isbn      = {978-3-540-29952-3},
  bibsource = {dblp computer science bibliography, http://dblp.org}
}

@article{DBLP:journals/talg/DellHMTW14,
  author       = {Holger Dell and
                  Thore Husfeldt and
                  D{\'{a}}niel Marx and
                  Nina Taslaman and
                  Martin Wahlen},
  title        = {Exponential Time Complexity of the Permanent and the {T}utte Polynomial},
  journal      = {{ACM} Trans. Algorithms},
  volume       = {10},
  number       = {4},
  pages        = {21:1--21:32},
  year         = {2014},
  url          = {https://doi.org/10.1145/2635812},
  doi          = {10.1145/2635812},
  bibsource    = {dblp computer science bibliography, https://dblp.org}
}

@inproceedings{DBLP:conf/icalp/BodlaenderH95,
  author       = {Hans L. Bodlaender and
                  Torben Hagerup},
  editor       = {Zolt{\'{a}}n F{\"{u}}l{\"{o}}p and
                  Ferenc G{\'{e}}cseg},
  title        = {Parallel Algorithms with Optimal Speedup for Bounded Treewidth},
  booktitle    = {Automata, Languages and Programming, 22nd International Colloquium,
                  ICALP95, Szeged, Hungary, July 10-14, 1995, Proceedings},
  series       = {Lecture Notes in Computer Science},
  volume       = {944},
  pages        = {268--279},
  publisher    = {Springer},
  year         = {1995},
  url          = {https://doi.org/10.1007/3-540-60084-1\_80},
  doi          = {10.1007/3-540-60084-1\_80},
  bibsource    = {dblp computer science bibliography, https://dblp.org}
}

@book{cygan2015parameterized,
  title={Parameterized Algorithms},
  author={Cygan, M. and Fomin, F.V. and Kowalik, {\L}. and Lokshtanov, D. and Marx, D. and Pilipczuk, M. and Pilipczuk, M. and Saurabh, S.},
  isbn={9783319212753},
  year={2015},
  publisher={Springer International Publishing}
}

@article{jtoran,
author = {Torán, Jacobo},
year = {1991},
month = {07},
pages = {753-774},
title = {Complexity Classes Defined by Counting Quantifiers.},
volume = {38},
journal = {J. ACM},
doi = {10.1145/116825.116858}
}

@inproceedings{DBLP:conf/coco/AllenderKRRV01,
  author       = {Eric Allender and
                  Michal Kouck{\'{y}} and
                  Detlef Ronneburger and
                  Sambuddha Roy and
                  V. Vinay},
  title        = {Time-Space Tradeoffs in the Counting Hierarchy},
  booktitle    = {Proceedings of the 16th Annual {IEEE} Conference on Computational
                  Complexity, Chicago, Illinois, USA, June 18-21, 2001},
  pages        = {295--302},
  publisher    = {{IEEE} Computer Society},
  year         = {2001},
  url          = {https://doi.org/10.1109/CCC.2001.933896},
  doi          = {10.1109/CCC.2001.933896},
  bibsource    = {dblp computer science bibliography, https://dblp.org}
}

@inproceedings{malod2007complexity,
  title={The complexity of polynomials and their coefficient functions},
  author={Malod, Guillaume},
  booktitle={Twenty-Second Annual IEEE Conference on Computational Complexity (CCC'07)},
  pages={193--204},
  year={2007},
  organization={IEEE}
}

@article{koiran2011interpolation,
  title={Interpolation in  {V}aliant’s theory},
  author={Koiran, Pascal and Perifel, Sylvain},
  journal={Computational Complexity},
  volume={20},
  pages={1--20},
  year={2011},
  publisher={Springer}
}

@article{DBLP:journals/acta/Wagner86,
  author       = {Klaus W. Wagner},
  title        = {The Complexity of Combinatorial Problems with Succinct Input Representation},
  journal      = {Acta Informatica},
  volume       = {23},
  number       = {3},
  pages        = {325--356},
  year         = {1986},
  url          = {https://doi.org/10.1007/BF00289117},
  doi          = {10.1007/BF00289117},
  bibsource    = {dblp computer science bibliography, https://dblp.org}
}

@article{DBLP:journals/eatcs/FlumG04,
  author       = {J{\"{o}}rg Flum and
                  Martin Grohe},
  title        = {Parametrized Complexity and Subexponential Time (Column: Computational
                  Complexity)},
  journal      = {Bull. {EATCS}},
  volume       = {84},
  pages        = {71--100},
  year         = {2004},
  bibsource    = {dblp computer science bibliography, https://dblp.org}
}

@article{GLYNN20101887,
title = {The permanent of a square matrix},
journal = {European Journal of Combinatorics},
volume = {31},
number = {7},
pages = {1887-1891},
year = {2010},
issn = {0195-6698},
doi = {https://doi.org/10.1016/j.ejc.2010.01.010},
url = {https://www.sciencedirect.com/science/article/pii/S0195669810000211},
author = {David G. Glynn}
}

@book{ryser,
author = {Ryser, Herbert John},
year = 1963, 
title = {Combinatorial Mathematics}, 
series = {Carus Mathematical Monographs},
volume = 14,
publisher = {Mathematical Association of America}
}

@article{shub1995intractability,
  title={{On the intractability of {H}ilbert’s {N}ullstellensatz and an algebraic version of “ {NP}$\neq$ {P}?”}},
  author={Shub, Michael and Smale, Steve},
  journal={Duke Mathematical Journal},
  volume={81},
  number={1},
  pages={47--54},
  year={1995},
  publisher={Duke University Press}
}

@article{blum1989theory,
  title={{On a theory of computation and complexity over the real numbers: {NP}-completeness, recursive functions and universal machines}},
  author={Blum, Lenore and Shub, Mike and Smale, Steve},
  journal={Bulletin (New Series) of the American Mathematical Society},
  volume={21},
  number={1},
  pages={1--46},
  year={1989},
  publisher={American Mathematical Society}
}

@incollection{blum2000algebraic,
  title={{Algebraic settings for the problem “{P}$\ne$ {NP}?”}},
  author={Blum, Lenore and Cucker, Felipe and Shub, Mike and Smale, Steve},
  booktitle={The Collected Papers of Stephen Smale: Volume 3},
  pages={1540--1559},
  year={2000},
  publisher={World Scientific}
}

@inproceedings{DBLP:conf/innovations/Koiran11,
  author    = {Pascal Koiran},
  title     = {{Shallow circuits with high-powered inputs}},
  booktitle = {Innovations in Computer Science - {ICS}},
  year={2011}
}

@inproceedings{dutta2021real,
  title={Real $\tau$-Conjecture for Sum-of-Squares: A Unified Approach to Lower Bound and Derandomization},
  author={Dutta, Pranjal},
  booktitle={International Computer Science Symposium in Russia},
  pages={78--101},
  year={2021},
  organization={Springer}
}

@phdthesis{tavenas2014bornes,
  title={{Bornes inferieures et superieures dans les circuits arithmetiques}},
  author={Tavenas, S{\'e}bastien},
  school={ Ecole Normale Supérieure de Lyon},
  year={2014}
}

@article{koiran2005valiant,
  title={Valiant’s model and the cost of computing integers},
  author={Koiran, Pascal},
  journal={computational complexity},
  volume={13},
  pages={131--146},
  year={2005},
  publisher={Springer}
}

@article{jukna2016optimality,
  title={On the optimality of Bellman--Ford--Moore shortest path algorithm},
  author={Jukna, Stasys and Schnitger, Georg},
  journal={Theoretical Computer Science},
  volume={628},
  pages={101--109},
  year={2016},
  publisher={Elsevier}
}

@article{DBLP:journals/corr/abs-1911-06738,
  author       = {Yaroslav Alekseev and
                  Dima Grigoriev and
                  Edward A. Hirsch and
                  Iddo Tzameret},
  title        = {Semi-Algebraic Proofs, {IPS} Lower Bounds and the {\(\tau\)}-Conjecture:
                  Can a Natural Number be Negative?},
  journal      = {CoRR},
  volume       = {abs/1911.06738},
  year         = {2019},
  url          = {http://arxiv.org/abs/1911.06738},
  eprinttype    = {arXiv},
  eprint       = {1911.06738},
  bibsource    = {dblp computer science bibliography, https://dblp.org}
}

\end{document}